\documentclass[12pt,draftcls,onecolumn]{IEEEtran}
\usepackage{cite}
\usepackage{amsmath,amssymb,amsfonts,mathtools}
\usepackage{algorithmic}
\usepackage{graphicx}
\usepackage{algorithm,algorithmic}
\usepackage{bm}
\usepackage{multicol}
\usepackage{lipsum}
\usepackage{float}
\usepackage{subfigure}
\usepackage{textcomp}

\usepackage{amsthm}
\usepackage{url}
\usepackage{caption} 
\newtheorem{theorem}{Theorem}

\newtheorem{lemma}{Lemma}[theorem]
\newtheorem{corollary}{Corollary}[theorem]
\newtheorem{prop}{Proposition}

\usepackage{hyperref}
\hypersetup{hidelinks=true}
\usepackage{textcomp}
\def\BibTeX{{\rm B\kern-.05em{\sc i\kern-.025em b}\kern-.08em
    T\kern-.1667em\lower.7ex\hbox{E}\kern-.125emX}}

\begin{document}
\title{Uncertainty Propagation on Unimodular Matrix\\ Lie Groups}
\author{Jikai Ye, Amitesh S. Jayaraman, and Gregory S. Chirikjian, \IEEEmembership{Fellow, IEEE}
\thanks{This work was supported by NUS Startup  grants A-0009059-02-00 and A-0009059-03-00, National Research Foundation, Singapore, under its Medium Sized Centre Programme - Centre for Advanced Robotics Technology Innovation (CARTIN),  sub award A-0009428-08-00, and AME Programmatic Fund Project MARIO A-0008449-01-00.}
\thanks{J. Ye and G. S. Chirikjian are with the Department of Mechanical Engineering, National University of Singapore, Singapore (e-mail: \{jikai.ye, mpegre\}@nus.edu.sg).}
\thanks{A. S. Jayaraman is with the Department of Mechanical Engineering, Stanford University, CA 94305, USA (e-mail: amiteshs@stanford.edu).}
}
\bibliographystyle{IEEEtran}

\maketitle

\begin{abstract}
This paper addresses uncertainty propagation on unimodular matrix Lie groups that have a surjective exponential map.
We derive the exact formula for the propagation of mean and covariance in a continuous-time setting from the governing Fokker-Planck equation.
Two approximate propagation methods are discussed based on the exact formula.
One uses numerical quadrature and another utilizes the expansion of moments.
A closed-form second-order propagation formula is derived.
We apply the general theory to the joint attitude and angular momentum uncertainty propagation problem and numerical experiments demonstrate two approximation methods.
These results show that our new methods have high accuracy while being computationally efficient.
\end{abstract}

\begin{IEEEkeywords}
Uncertainty propagation, matrix Lie group, Fokker-Planck equation, attitude estimation
\end{IEEEkeywords}
\section{Introduction}\label{sec:Intro}

\IEEEPARstart{T}{he} problem of uncertainty propagation is a classical topic in robotics and control.
For a stochastic dynamical system, the goal is to estimate the evolution of the probability distribution of the state variables.
When the configuration space of a system has a Lie group structure, such as the rotation group $SO(3)$ for a three-dimensional rigid body, the uncertainty propagation becomes more challenging due to the non-Euclidean nature of the group manifold.
This paper provides a precise description of the propagation process of mean and covariance for a class of unimodular matrix Lie groups.

Uncertainty propagation is usually seen in filtering algorithms as the first step.
Traditional filters for nonlinear systems in Euclidean space, such as the extended Kalman filter (EKF) and unscented Kalman filter (UKF), both have their counterparts in Lie groups.
The invariant extended Kalman filter (IEKF), compared to a parametrized filter, has a nice convergence property \cite{barrau2018invariant}.
For a continuous dynamical system, IEKF calculates the deterministic trajectory as the propagation of the mean and linearizes the transition model on the exponential coordinate to propagate the covariance \cite{barrau2016invariant}.
The original derivation for continuous systems is simplified in \cite{phogat2020invariant}.
In \cite{bourmaud2015continuous}, the authors derive a first-order propagation formula from a stochastic differential equation and justify that the approximation error of the mean remains zero up to first-order.
Note that all the EKFs on Lie groups use the deterministic trajectory for mean propagation and a linearized model for propagating the covariance.
The UKF on Lie groups proposed in \cite{brossard2017unscented} saves the linearization step.
It propagates the mean using the deterministic transition model and calculates the propagation of sigma points by an exponential map.
The covariance is approximated by the covariance of the propagated sigma points on the exponential coordinate.
Later, the authors generalize this UKF to parallelizable manifolds in \cite{brossard2020code}.
In \cite{loianno2016visual}, the mean propagation is first calculated by the deterministic transition model and then adjusted by the algorithmic average of sigma points on the exponential coordinate, while in \cite{forbes2017sigma} the mean is adjusted by optimization on the Lie group.
The paper \cite{sjoberg2021lie} does not use the exponential map for sigma points propagation, but uses the Jacobian matrix to formulate the transition model on the exponential coordinate.
The Gaussian constructed by the mean and covariance on the coordinate is mapped back to a concentrated Gaussian on the group after each step.

All previous methods start from the transition model and perform approximations based on it.
In contrast, our method first converts the noisy transition model to a corresponding Fokker-Planck equation (FPE) and extracts the exact propagation of the mean and covariance from the FPE.
Also, we treat the mean propagation by a change of variables on Lie groups, which avoids the artifact introduced by choosing a specific coordinate system or the time-consuming optimization step used in previous works.
The most related work is \cite{jayaraman2023lietheoretic}, which develops the propagation formula for the cotangent bundle group of $SO(3)$.
Here we provide a general formula for a class of unimodular Lie groups with a computable exact propagation formula and two approximation methods.


When a Fokker-Planck equation is formulated in Euclidean space, several approximation methods have been proposed, like the small noise expansion \cite{gardiner1985handbook}, homotopy perturbation method \cite{biazar2008homotopy}, Adomian method \cite{tatari2007application}, variational integration method \cite{biazar2010variational}, new iterative method \cite{hemeda2018new}, etc.
All these methods recursively compute an approximate solution that approaches the exact solution with increasing iterations.
When the domain of the FPE is the whole Euclidean space, the recursive equations for the probability density function can be converted to a set of ordinary differential equations for the moments by using integration by parts.
The method that is the most similar to ours is the small noise expansion in \cite{gardiner1985handbook}, which recenters the probability density function at the solution to the deterministic trajectory, expands the distribution into an infinite sum of functions weighted by the power of a small noise coefficient, and matches terms with the same weight to form the set of equations.
Our approach recenters around the mean, not the deterministic solution, derives the exact propagation formula to approximate instead of approximating from the start, and avoids the complex procedure to match terms in the perturbation method.
Also, since the nice property of partial derivative in Euclidean space, $\frac{\partial x_i}{\partial x_j}=\text{const}$, does not hold in a Lie group setting, \textit{i}. \textit{e}. $E^r_i(x_j)\neq \text{const}$, where $E^r_i$ is the Lie derivative in the $i$th basis direction and $x_j$ is the $j$th coordinate variable in the exponential coordinates, the Jacobian matrix will show up in calculating these terms and makes the calculation more complicated.

This paper contributes to providing an exact propagation formula for the mean and covariance using the Fokker-Planck equation corresponding to the noisy transition model.
Two approximate propagation methods based on numerical quadrature and expansion of moments are discussed.
For the expansion of moments, we also provide a closed-form second-order propagation formula for mean and covariance.
Both methods can be used for unimodular Lie groups that have a surjective exponential map, such as the rotation group $SO(N)$, the special Euclidean group $SE(N)$, and their tangent/cotangent bundle groups \cite{ng2022equivariant,jayaraman2020black,jayaraman2023inertial}.
We apply the theory to propagating the uncertainty of the attitude and angular momentum of a rigid body whose motion is governed by Euler's equation of motion.
A formula based on the unscented transform and a closed-form second-order propagation formula based on the expansion of moments are implemented and compared numerically to an existing unscented Kalman filter \cite{sjoberg2021lie}.
Both formulas achieve higher accuracy compared to the baseline method and the closed-form formula is much more computationally efficient.

The paper is organized as follows:
the background knowledge of Lie group theory and Fokker-Planck equation is introduced in Section \ref{sec:LieGrpTheory}; 
the general theory of error propagation on Lie group, approximations based on unscented transform and expansion of moments, and their connection to the propagation step of Kalman filters are presented in Section \ref{sec:EoM_FPE};
the rotational dynamics, the direct product group structure on the phase space, and the propagation formula are stated in Section \ref{sec:GoverningEquations};
the numerical experiments and results are in Section \ref{sec:Results}; the paper concludes in Section \ref{sec:Concl}.


\section{\label{sec:LieGrpTheory} Background}

In this section, we introduce the underlying concepts and results which will be used in the following section.
Section \ref{sec:back_A} reviews the basics of Lie group theory and Section \ref{sec:back_B} reviews the stochastic differential equations (SDE) on a Lie group and its corresponding Fokker-Planck equation on the group.
A more comprehensive introduction to these materials can be found in \cite{chirikjian2016harmonic,chirikjian2011stochastic}.

\subsection{Lie group basics} \label{sec:back_A}

A \emph{group} is a non-empty set together with a closed binary operation. 
That is, for two elements, $g_1$ and $g_2$ in a group $G$, there exists a known operation $\circ$ such that $g_1\circ g_2 \in G$. 
For a group, this operation must be \emph{associative} so that, $g_1\circ(g_2\circ g_3) = (g_1\circ g_2)\circ g_3$. Furthermore, a group always has an \emph{identity} element such that for every $g\in G$, we have an $e$ so that $e\circ g = g \circ e = g$. 
Finally, every element in a group must have an \emph{inverse} so that if $g^{-1}\circ g = g\circ g^{-1} = e$ then $g^{-1}$ is the inverse of $g$ and $g^{-1}\in G$.
One example of a group is the real general linear group $GL(n,\mathbb{R})$, which is the collection of all invertible $n\times n$ real matrices together with the binary operation being matrix multiplication.
A \emph{subgroup} $H$ of a group $G$ is a subset of $G$ with the same binary operation closed in $H$.

A \emph{Lie group} is a group that is also an analytic manifold and for which the group operations of multiplication and inverse of elements are analytic.
We only consider \emph{real matrix Lie groups} in this paper, which are subgroups of $GL(n,\mathbb{R})$. 
The dimension of a matrix Lie group is defined by the dimension of the manifold, not the size of the matrix.
The \emph{Lie algebra} $\mathcal{G}$ of a matrix Lie group $G$ is a vector space of matrices with an additional closed binary operation called the Lie bracket: $[X, Y]\doteq XY-YX, \, \forall X,Y\in \mathcal{G}$.
The dimension of the Lie algebra is the same as the dimension of the Lie group.
After choosing a set of basis $\{E_1,\,E_2,\,\cdots,\,E_N\}$ for $\mathcal{G}$, each element of $\mathcal{G}$ can be written as the linear sum of them, \textit{i}. \text{e}. $X=\sum_{i=1}^N x_i E_i$.
We stack the weight, $x_i$, into a column vector, $\boldsymbol{x}$, and use the notation ``$\wedge$'' to convert it into the matrix form $X=\boldsymbol{x}^{\wedge}$.
The inverse operation is denoted as ``$\vee$'', which converts a matrix (an element of the Lie algebra) into a column vector.
Choosing a basis and declaring that their inner product is defined by $(E_i, \,E_j)=\boldsymbol{e}_i\cdot \boldsymbol{e}_j=\delta_{ij}$, where $\delta_{ij}=1$ when $i=j$ and $\delta_{ij}=0$ otherwise, is equivalent to fixing a metric tensor for $G$.

Now, we define the ``little ad'' operator. 
For $X,Y\in\mathcal{G}$, we define $ad(X)Y \doteq [X,Y]$. 
It can be seen that $ad(X)$ is a linear map on the Lie algebra, so it has a matrix form $[ad(X)]$ such that $[X,Y]^{\vee}= (ad(X)Y)^\vee=[ad(X)]Y^\vee$.
The matrix form of the operator can be calculated by
\begin{equation}\label{eq:littlead}
[ad(X)] = \left[[X,E_1]^\vee,\cdots,[X,E_N]^\vee\right],
\end{equation}
The adjoint representation, $Ad:G\rightarrow GL(\mathcal{G})$, is defined by: $Ad(g)X\doteq gXg^{-1}$, which also has a matrix form that obeys $[Ad(g)]X^{\vee}=(Ad(g)X)^{\vee}$.
The formula for $[Ad(g)]$ is
\begin{equation}
    [Ad(g)]=[(gE_1 g^{-1})^{\vee}, \,(gE_2 g^{-1})^{\vee},\,...,\,(gE_N g^{-1})^{\vee}].
\end{equation}

A matrix Lie group and its Lie algebra are related by matrix exponential and logarithm: 
\begin{equation}
    \exp(X)\doteq \sum_{k=0}^{\infty}\frac{1}{k!}X^k \in G, \,\, \text{if} \, X \in \mathcal{G},
\end{equation}
\begin{equation}
    \log(A)\doteq \sum_{k=1}^{\infty}\frac{(-1)^{k+1}}{k}(A-\mathbb{I})^{k} \in \mathcal{G}, \,\, \text{if} \, A \in {G}.
\end{equation}
The exponential map is bijective within a neighborhood of $X=\mathbb{O}$, where we can parametrize each group element as $g(\boldsymbol{x})=\exp(\boldsymbol{x}^{\wedge})$.
It can be considered as a coordinate system of the group around the identity,
which is called \emph{exponential coordinate}.
In this paper, we consider groups whose exponential map is surjective.
There may exist other parametrization methods, like the Euler angle for $SO(3)$.
We use $g(\boldsymbol{q})$ to denote a general parametrization and $g(\boldsymbol{x})$ to denote the exponential coordinate.

The right Jacobian for a parametrization \cite{chirikjian2016harmonic} is defined by the matrix
\begin{equation}\label{eq:JR}
J_r(\boldsymbol{q}) = \left[\left(g^{-1}\frac{\partial g}{\partial q_1}\right)^\vee,\cdots,\left(g^{-1}\frac{\partial g}{\partial q_N}\right)^\vee\right],
\end{equation}
and the left Jacobian matrix is
\begin{equation}\label{eq:JL}
J_l(\boldsymbol{q}) = \left[\left(\frac{\partial g}{\partial q_1}g^{-1}\right)^\vee,\cdots,\left(\frac{\partial g}{\partial q_N}g^{-1}\right)^\vee\right],
\end{equation}
where in both cases the $[\cdots]$ emphasises that this is a matrix. 
Note that $J_r\neq J_l$ in general and that $[Ad(g)]=J_l J_r^{-1}$. 
Further, note that $J_r(h\circ g(\boldsymbol{q})) = J_r(g(\boldsymbol{q}))$ and $J_l(g(\boldsymbol{q})\circ h) = J_l(g(\boldsymbol{q}))$ for a constant $h\in G$. 
That is, $l$ and $r$ denote on which side the partial derivatives appear in the columns of the Jacobian, which is opposite to the invariance. 
The Jacobians connect the derivative on matrices to the derivative on coordinates:
\begin{equation}
    \left [ g^{-1}(\boldsymbol{q}(t))\frac{d g(\boldsymbol{q}(t))}{dt}\right ]^{\vee}=J_r(\boldsymbol{q}(t))\frac{d \boldsymbol{q}(t)}{dt},
\end{equation}
\begin{equation} \label{eq:l_J_para}
    \left [ \frac{d g(\boldsymbol{q}(t))}{dt} g^{-1}(\boldsymbol{q}(t)) \right ]^{\vee}=J_l(\boldsymbol{q}(t))\frac{d \boldsymbol{q}(t)}{dt}.
\end{equation}

The right and left Lie derivatives of the group are defined for a differentiable function $f:G\rightarrow\mathbb{R}^m$ to be
\begin{equation}\label{eq:rightDerivativeNP}
    (E^r_i f)(g) \;\doteq\; \left.\frac{d}{dt}\right|_{t = 0}f(g\circ \exp({tE_i})),
\end{equation}
\begin{equation}\label{eq:leftDerivativeNP}
    (E^l_i f)(g) \;\doteq\; \left.\frac{d}{dt}\right|_{t = 0}f(\exp(-tE_i)\circ g),
\end{equation}
which is akin to the concept of a directional derivative on the group.
In this convention, $l$ and $r$ refer to which side of the argument $\exp(t E_i)$ appears not the invariance of the operators, which is the opposite.
That is, $E^l$ is right invariant and $E^r$ is left invariant.
It can be shown \cite{chirikjian2016harmonic} that the parametric form of the right and left Lie derivatives can also be expressed in terms of Jacobians as
\begin{equation}\label{eq:rightDerivativeJT}
    (E^r_i f)(g(\boldsymbol{q})) = \sum_{j=1}^N [J_r(\boldsymbol{q})]^{-T}_{ij}\frac{\partial f(g(\boldsymbol{q}))}{\partial q_j},
\end{equation}
\begin{equation}\label{eq:leftDerivativeJT}
    (E^l_i f)(g(\boldsymbol{q})) = -\sum_{j=1}^N [J_l(\boldsymbol{q})]^{-T}_{ij}\frac{\partial f(g(\boldsymbol{q}))}{\partial q_j}.
\end{equation}

We proceed to introduce the Haar measure and integration on a Lie group.
There are two such measures (each of which is unique up to scaling), $d_l g \doteq |\det J_l(\boldsymbol{q})|d\boldsymbol{q}$ and $d_lg \doteq |\det J_l(\boldsymbol{q})|d\boldsymbol{q}$. 
If $|\det [Ad(g)]| = |\det J_l|/|\det J_r|=1$, the group is said to be \emph{unimodular}.
For a constant element $h\in G$, we have $d_l(g\circ h)=d_l g=d_r g = d_r(h\circ g)$ for a unimodular group.
If the domain $D\subseteq \mathbb{R}^N$ of the parametrization has a one-to-one correspondence to the group except for a set of measure zero in $G$, we call the parametrization ``almost global'', and the integration on the group is defined by the integration on $D$ as
\begin{equation}
    \int_G f(g)dg=\int_{D\subseteq \mathbb{R}^N} f(g(\boldsymbol{q}))|\det{J_l(\boldsymbol{q})}|d\boldsymbol{q}.
\end{equation}
The groups we consider in this paper all have surjective exponential maps and thus the exponential coordinate is such a parametrization.

A probability density function $f(g)$ on $G$ is a function that satisfies the properties that $f(g) \geq 0$ for all $g\in G$ and $$
\int_G\;f(g)\;dg = 1.$$
The ``right'' group-theoretic mean $\mu_r$ and covariance $\Sigma_r$ of $f(g,t)$ are defined by:
\begin{equation}\label{eq:GroupMu}
\int_G \log^\vee(g \circ \mu_r^{-1})f(g)dg \;\doteq\; \boldsymbol{0},
\end{equation}
and,
\begin{equation}\label{eq:GroupCov}
\Sigma_r \;\doteq\; \int_G [\log^\vee(g \circ \mu_r^{-1})][\log^\vee(g \circ \mu_r^{-1})]^T f(g)dg.
\end{equation}
The definition of ``left'' group-theoretic mean $\mu_l$ and covariance $\Sigma_l$ differs by placing the inverse of the mean on the left of $g$ in the integration.
It can be shown that these two means have the same value and the covariances are related by an adjoint representation: $\Sigma_r=[Ad(\mu)]\Sigma_l [Ad(\mu)]^T$.


An important unimodular group is the special orthogonal group, $SO(3)$.
The basis of Lie algebra we use are
$$
  \setlength{\arraycolsep}{2pt}
  \renewcommand{\arraystretch}{1}
{E}_1=\begin{pmatrix}
0 & 0 & 0 \\
0 & 0 & -1\\
0 & 1 & 0
\end{pmatrix}, \,
{E}_2=\begin{pmatrix}
0 & 0 & 1 \\
0 & 0 & 0\\
-1 & 0 & 0
\end{pmatrix}, \\ \,
{E}_3=\begin{pmatrix}
0 & -1 & 0 \\
1 & 0 & 0\\
0 & 0 & 0
\end{pmatrix}.
$$
The closed-form expressions of the exponential, logarithm, right, and left Jacobian that uses this basis can be found in \cite{chirikjian2011stochastic}.
The exponential coordinate of $SO(3)$ is almost global and its domain is a solid ball $D=\{ \boldsymbol{x} \in \mathbb{R}^3 \, | \, \lVert\boldsymbol{x}\rVert_2<\pi \}$.
The inverse of the left/right Jacobian also has a closed-form expression \cite{chirikjian2011stochastic,park1991optimal}:
\begin{equation}\label{eq:left_Jacobian_SO3}
J_l^{-1}(\boldsymbol{x})=\mathbb{I}-\frac{1}{2}X+\left( \frac{1}{||\boldsymbol{x}||}-\frac{1+\cos ||\boldsymbol{x}||}{2||\boldsymbol{x}|| \sin ||\boldsymbol{x}||} \right )X^2,
\end{equation}
\begin{equation}\label{eq:right_Jacobian_SO3}
J_r^{-1}(\boldsymbol{x})=\mathbb{I}+\frac{1}{2}X+\left( \frac{1}{||\boldsymbol{x}||}-\frac{1+\cos ||\boldsymbol{x}||}{2||\boldsymbol{x}|| \sin ||\boldsymbol{x}||} \right )X^2,   
\end{equation}
which will be used in the following sections.



We close this section by introducing \emph{direct product group}, which is a group constructed by two groups.
For groups $(G,\circ)$ and $(\mathbb{R}^N,+)$, we define the it to be $(K,\bullet)\doteq (G,\circ)\times(\mathbb{R}^N,+)$. 
Each element $k\in K$ can be written as $k=(g,\boldsymbol{v})$ where $g\in G$ and $\boldsymbol{v}\in R^N$.
The group operation for $K$ is defined by
$$
k_1 \bullet k_2 = (g_1,\boldsymbol{v}_1)\bullet (g_2,\boldsymbol{v}_2)\doteq (g_1\circ g_2,\boldsymbol{v}_1+\boldsymbol{v}_2).
$$
It can be understood as performing products separately for elements from different groups.
When $G$ has a matrix representation, we can represent each group element of $K$ by matrix and perform abstract group operations by matrix multiplication:
$$
\begin{aligned}
[k_1][k_2]&\doteq \begin{pmatrix}
    [g_1] & \mathbb{O} \\
    \mathbb{O} & \text{diag}(e^{\boldsymbol{v_1}})
\end{pmatrix}
\begin{pmatrix}
    [g_2] & \mathbb{O} \\
    \mathbb{O} & \text{diag}(e^{\boldsymbol{v_2}})
\end{pmatrix}\\
&= 
\begin{pmatrix}
    [g_1][g_2] & \mathbb{O} \\
    \mathbb{O} & \text{diag}(e^{\boldsymbol{v_1}}) \text{diag}(e^{\boldsymbol{v_2}})
\end{pmatrix}\\
&=
\begin{pmatrix}
    [g_1\circ g_2] & \mathbb{O} \\
    \mathbb{O} & \text{diag}(e^{\boldsymbol{v_1}+\boldsymbol{v_2}})
\end{pmatrix}\\
&=[k_1 \bullet k_2],
\end{aligned}
$$
where $[\cdot]$ means matrix representation and $\text{diag}(e^{\boldsymbol{v}})$ stands for a diagonal matrix whose $i$th diagonal element is $e^{v_i}$.


\subsection{Converting stochastic differential equations to Fokker-Planck equations on unimodular groups}\label{sec:back_B}
In this section, we consider a unimodular group $G$ with Lie algebra $\mathcal{G}$. Assume that a group element $g\in G$ can be parameterized by $g = g(\boldsymbol{q})$ where $\boldsymbol{q}\in\mathbb{R}^N$. 

A ``right'' Stratonovich stochastic differential equation 
has the form:
\begin{equation}\label{eq:GrpStraton}
    (g^{-1}dg)^\vee =\boldsymbol{h}(g(\boldsymbol{q}),t)dt + H(g(\boldsymbol{q}),t)\;\circledS\; d\boldsymbol{W},
\end{equation}
where $(g^{-1}dg)^\vee \doteq J_r(\boldsymbol{q})d\boldsymbol{q}$.
The $\circledS$ symbol denotes that the equation is interpreted in the Stratonovich sense, which is used for physical systems and is different from the It\^{o}'s SDE \cite{gardiner1985handbook}.
Then, $d\boldsymbol{W}$ is a Wiener white noise increment, which satisfies $\langle d\boldsymbol{W}\rangle = \boldsymbol{0}$ and $\langle dW_idW_j\rangle = \delta_{ij}dt$. Similarly, a ``left'' stochastic differential equation has the form:
\begin{equation}\label{eq:LeftGrpStraton}
    (dg\;g^{-1})^\vee = \boldsymbol{h}(g(\boldsymbol{q}),t)dt + H(g(\boldsymbol{q}),t)\;\circledS\; d\boldsymbol{W}.
\end{equation}
We now present a theorem that relates the Fokker-Planck equation evolving on $G$ with the right Stratonovich stochastic differential equation.
In these theorems as well as in the results presented in subsequent sections, Einstein summation notation is employed.
That is $v_i w_i$ is a shorthand for $\sum_i v_i w_i$.

\begin{theorem}\label{thm:Thm1}
The Fokker-Planck equation for the probability density function $f(g,t)$ corresponding to the right Stratonovich stochastic differential equation of the form in (\ref{eq:GrpStraton}) evolves on the unimodular group $G$ as:
\begin{equation}\label{eq:rightFPE_correct_statement}
    \frac{\partial f}{\partial t} +  E^r_m\left(\left[h_m +\frac{1}{2}{H}_{nj}E^r_n(H_{mj}) \right]f\right) =\frac{1}{2}E^r_mE^r_p(H_{mn}H^T_{np}f),
\end{equation}
where $E^r_i$ is the right Lie derivative in (\ref{eq:rightDerivativeNP}). 
\end{theorem}
\begin{proof}
The proof is in Appendix \ref{proof:FPE}.
\end{proof}

\begin{corollary}\label{crlly:Corollary1}
If the noise coefficient matrix $H$ only depends on time, \textit{i}. \textit{e}. $H=H(t)$, the right Fokker-Planck equation for the right Stratonovich stochastic differential equation (\ref{eq:GrpStraton}) is
\begin{equation}\label{eq:rightFPE_correct_statement_simpl}
    \frac{\partial f}{\partial t} +  E^r_m\left(h_mf\right) =\frac{1}{2}E^r_mE^r_l(H_{mn}H^T_{nl}f).
\end{equation}
\end{corollary}

\begin{proof}
When $H$ is not a function of $g(\bm{q})$, we have $\partial {H_{mj}(t)}/\partial q_k = 0$. 
The calculation of the Lie derivative $E_n^r(H_{mj})$ in Theorem \ref{thm:Thm1} using the Jacobian matrix indicates the result.
\end{proof}
For the left Fokker-Plack equation, we have a similar result:
\begin{theorem}\label{thm:Thm2}
The Fokker-Planck equation for the probability density function $f(g,t)$ corresponding to the left Stratonovich stochastic differential equation of the form in (\ref{eq:LeftGrpStraton}) evolves on the unimodular group $G$ as:
\begin{equation}\label{eq:leftFPE_correct}
    \frac{\partial f}{\partial t} -  E^l_m\left(\left[h_m + \frac{1}{2}{H}_{nj}E^l_n(H_{mj})\right]f\right) =\frac{1}{2}E^l_mE^l_p(H_{mn}H^T_{np}f),
\end{equation}
where $E^l_i$ is the left Lie derivative in (\ref{eq:leftDerivativeNP}).
\end{theorem}
\begin{proof}
The proof follows the same procedure as that for Theorem \ref{thm:Thm1}.
\end{proof}
\begin{corollary}\label{crlly:Corollary2}
If the noise coefficient matrix $H$ only depends on time, the left Fokker-Planck equation for the left Stratonovich stochastic differential equation (\ref{eq:LeftGrpStraton}) is
\begin{equation}\label{eq:leftFPE_correct_statement_simpl}
    \frac{\partial f}{\partial t} -  E^l_m\left(h_mf\right) =\frac{1}{2}E^l_mE^l_p(H_{mn}H^T_{np}f).
\end{equation}
\end{corollary}
\begin{proof}
This can be shown using an argument similar to that in Corollary \ref{crlly:Corollary1}.
\end{proof}


\section{Propagating Uncertainty by Fokker-Planck Equation on Unimodular Lie Groups}\label{sec:EoM_FPE}

In this section, we present a general theory that describes the propagation of the mean and covariance of the solution to a Fokker-Planck equation (FPE) on unimodular Lie groups that have a surjective exponential map.
An exact propagation formula for a left FPE is derived in Section \ref{sec:general_EOM}.
A Gaussian distribution on the group is constructed using the propagated mean and covariance and can serve as an approximate solution to the FPE.
Two approximate propagation methods that employ numerical quadrature and expansion of moments are discussed with a closed-form second-order propagation formula presented.
In section \ref{sec:FPE_connection}, we state the connection between the solution to a right and a left FPE, which can be used to convert the propagation equations for a left equation to a right equation.
We show that the propagation step of traditional Kalman filters also falls within the umbrella of our theory in Section \ref{sec:EKF}.

\subsection{Propagation equations for left Fokker-Planck equation} \label{sec:general_EOM}
Suppose we have a left stochastic differential equation (\ref{eq:LeftGrpStraton}) and its corresponding left Fokker-Planck equation on a matrix Lie group $G$,
\begin{equation}\label{eq:leftFPE_EOM}
\left \{
    \begin{aligned}
    &\frac{\partial f}{\partial t} \,=E^l_i\left(h_i f\right) + \frac{1}{2}E^l_iE^l_j(H_{ik}H^T_{kj}f)\\
    &f(g,0)=\delta\left(g\right)
    \end{aligned}
    \right.
\end{equation}
where we assume $h_i(g,t)$ is a function of $G$ and time, $H_{ij}(t)$ is a function of time.
The Dirac delta function $\delta(g)$ is defined by
\begin{equation}
    \int_G \delta(g)dg=1 \quad\text{and}\quad\int_G \varphi(g)\delta(g)dg=\varphi(e),
\end{equation}
where $\varphi(g)$ is an arbitrary well-behaved function and $e$ is the identity element of the group.
We want to extract the right group-theoretic mean $\mu(t)$ and covariance $\Sigma(t)$ of the distribution $f(g,t)$ as defined in (\ref{eq:GroupMu}) and (\ref{eq:GroupCov}).

The following steps summarize the procedure for deriving the propagation equations: 
(1) Re-center the distribution to the mean by change of variable, 
(2) Use integration by parts to derive the exact propagation formula,
(3) Exploit quadrature methods or expansion of moments to give approximate formulas.

\subsubsection{Re-centering} \label{sec:recenter}
In this step, we define a new distribution $\rho(k,t)\doteq f(k\circ \mu(t),t)$, for which the right group-theoretic mean is always the identity of the group, \textit{i}. \textit{e}.
\begin{equation} \label{eq:rho_mean}
\int_G \log^{\vee} (k) \rho(k)dk \equiv \boldsymbol{0}.
\end{equation}
This technique is adopted by many previous works for propagating group-theoretic mean and covariance \cite{wang2006error,long2013banana,park2010path,jayaraman2023lietheoretic}.
We can reformulate the Fokker-Planck equation (\ref{eq:leftFPE_EOM}) by $\rho(k,t)$ using the following result:
\begin{theorem} \label{thm:change_var}
If $f(g,t)$ is the solution to (\ref{eq:leftFPE_EOM}) and $\rho(k,t)\doteq f(k\circ \mu(t),t)$, where $\mu(t)$ is the right group-theoretic mean of $f(g,t)$, then $\rho(k,t)$ satisfies the following equation:
\begin{equation} \label{eq:FPE_rho}
    \frac{\partial \rho}{\partial t}=(\dot \mu \mu^{-1})^{\vee}_i \cdot E^r_i\rho+E_i^l({h}^c_i \rho)+\frac{1}{2}E_i^l E_j^l ({H}_{ik}{H}^T_{kj} \rho),
\end{equation}    
where $\boldsymbol{h}^c(k,t)\doteq \boldsymbol{h}(k\circ \mu(t),t)$.
\end{theorem}
\begin{proof}
    The proof is in Appendix \ref{proof:change_var}.
\end{proof}

\subsubsection{Integration by parts}
Since the mean of $\rho(k)$ is always the identity, we have the following equation:
$$
\int_G \boldsymbol{x}(k) \frac{\partial \rho(k,t)}{\partial t} dk=\frac{\partial}{\partial t}\int_G \boldsymbol{x} (k)\rho(k,t)dk=0,
$$
where $\boldsymbol{x}(k)=\log^{\vee} k$.
Multiplying both sides of (\ref{eq:FPE_rho}) by $\boldsymbol{x}(k)$ and integrating on $G$ yields
\begin{equation}
    \boldsymbol{0}\!=\!\int_G \! \boldsymbol{x} \!\left( \!(\dot \mu \mu^{-1})^{\vee}_i \! \cdot \! E^r_i\rho+E_i^l(h^c_i \rho)\!+\!\frac{1}{2} E_i^l E_j^l ({H}_{ik}{H}^T_{kj} \rho)\!\right)\! dk.
\end{equation}
When the function $\rho(k)$ is absolutely integrable, integration by parts \cite{chirikjian2011stochastic} can be performed as
\begin{equation}
    \begin{aligned}
    \boldsymbol{0}=-\int_G E^r_i(\boldsymbol{x}) \rho(k)dk \cdot \boldsymbol{e}_i^T (\dot \mu \mu^{-1})^{\vee}-\int_G E_i^l(\boldsymbol{x})({h}_i^c \rho)dk+\frac{1}{2}\int_G E_j^l E_i^l (\boldsymbol{x}) ({H}_{ik}{H}^T_{kj} \rho)dk.
    \end{aligned}
\end{equation}
Therefore, we have
\begin{equation}
    \begin{aligned}
    (\dot \mu \mu^{-1})^{\vee}=\langle E^r_i(\boldsymbol{x})\boldsymbol{e}_i^T  \rangle^{-1} \langle-{h}^c_i E_i^l(\boldsymbol{x}) +\frac{1}{2}{H}_{ik}{H}^T_{kj} E_j^l E_i^l (\boldsymbol{x})\rangle,
    \end{aligned}
\end{equation}
where $\langle \varphi\rangle \doteq \int_G \varphi(k)\rho(k,t)dk$.
The invertibility of $\langle E^r_i(\boldsymbol{x})\boldsymbol{e}_i^T  \rangle$ is guaranteed for a short time, which will be explained later.
The same calculation can be performed to obtain the propagation equation for covariance and higher moments.
The results are summarized as:
\begin{theorem} \label{thm:prop_all_moments}
For a short time, the mean, covariance, and higher moments of the distribution $\rho(k,t)$ in (\ref{eq:FPE_rho}) obey the following ordinary differential equations:
\begin{equation} \label{eq:exact_EOM}
\left \{
\begin{aligned}
       &(\dot \mu \mu^{-1})^{\vee}=\langle E^r_i(\boldsymbol{x})\boldsymbol{e}_i^T  \rangle^{-1} \langle \frac{1}{2}{H}_{ik}{H}^T_{kj} E_j^l E_i^l (\boldsymbol{x}) -{h}^c_i E_i^l(\boldsymbol{x})\rangle\\&
       \frac{d \langle x_m \rangle}{dt}=0\\
       &\frac{d\langle x_m x_n \rangle}{dt}= \langle \frac{1}{2}{H}_{ik}{H}^T_{kj} E^l_j E^l_i(x_m x_n )\!-\!(\dot{\mu}\mu^{-1})^{\vee}_i E^r_i(x_m x_n)\\
        & \qquad\qquad\qquad\qquad-{h}^c_i E^l_i(x_m x_n)\rangle \\
         &\qquad \vdots\\
        &\frac{d \langle x_{i_1}x_{i_2}\cdots x_{i_n}\rangle}{dt}  = \langle \frac{1}{2}{H}_{ik}{H}^T_{kj} E^l_j E^l_i(x_{i_1}x_{i_2}\cdots x_{i_n}) \\ &\qquad\qquad \qquad\qquad -(\dot{\mu}\mu^{-1})^{\vee}_i E^r_i(x_{i_1}x_{i_2}\cdots x_{i_n})\\
        &\qquad\qquad \qquad \qquad -{h}^c_i E^l_i(x_{i_1}x_{i_2}\cdots x_{i_n})\rangle \\
        &\qquad \vdots
\end{aligned}
\right.
\end{equation}
\end{theorem}
In principle, the propagation of all moments $\langle x_{i_1}x_{i_2}...x_{i_n}\rangle$ can be obtained.
However, for a long time, the higher-order moments may go to infinity, which makes the propagation equations ill-defined.
So we only consider propagating for a short time.

To compute the RHS of (\ref{eq:exact_EOM}), a parametrized formula is needed.
The following lemma provides a way to compute the Lie derivative of $\boldsymbol{x}(k)$:
\begin{lemma} \label{lemma:lie_J}
    When using exponential coordinates to parameterize the group, the left and right Lie derivatives of $\boldsymbol{x}(k)=\log^{\vee}k$ are:
    \begin{equation}
        E^l_i(\boldsymbol{x})=-J_l^{-1}(\boldsymbol{x})\boldsymbol{e}_i,
    \end{equation}
    \begin{equation} E^r_i(\boldsymbol{x})=J_r^{-1}(\boldsymbol{x})\boldsymbol{e}_i,
    \end{equation}
    and
    \begin{equation}
        E^l_i E^l_j(\boldsymbol{x})=[J_l^{-T}(\boldsymbol{x})]_{ik}\frac{\partial J_l^{-1}(\boldsymbol{x})}{\partial x_k}\boldsymbol{e}_j,
    \end{equation}
    \begin{equation} E^r_i E^r_j(\boldsymbol{x})=[J_r^{-T}(\boldsymbol{x})]_{ik}\frac{\partial J_r^{-1}(\boldsymbol{x})}{\partial x_k}\boldsymbol{e}_j.
    \end{equation}
    \begin{proof}
        Since $\boldsymbol{x}=\log^{\vee}(k)$ and we parametrize the group using exponential coordinates, the ${f}(g(\boldsymbol{x}))$ in (\ref{eq:rightDerivativeJT}) is identified as ${f}(g(\boldsymbol{x}))=\boldsymbol{x}$. 
        Using (\ref{eq:rightDerivativeJT}), we have
        $$
        E^r_i(\boldsymbol{x})=[J_r]^{-T}_{ij}\frac{\partial}{\partial x_j} \boldsymbol{x}=J_r^{-1}\boldsymbol{e}_i.
        $$
        Other equations can be obtained through direct calculations.
    \end{proof}
\end{lemma}
The Lie derivative of $\boldsymbol{x}\boldsymbol{x}^T$ can be calculated by the Leibniz rule:
\begin{equation*}
    \begin{aligned}
&E_i^l(\boldsymbol{x}\boldsymbol{x}^T)=E_i^l(\boldsymbol{x})\boldsymbol{x}^T+\boldsymbol{x}E_i^l(\boldsymbol{x}^T),\\
    E_i^l &E_j^l(\boldsymbol{x}\boldsymbol{x}^T)=\text{sym}\left [E_i^l E_j^l(\boldsymbol{x})\boldsymbol{x}^T+E_i^l(\boldsymbol{x}) E_j^l(\boldsymbol{x}^T)\right ],
    \end{aligned}
\end{equation*}
where $\text{sym}(A)\doteq A+A^T$.
Substituting them back into (\ref{eq:exact_EOM}) and organizing terms, we have:
\begin{corollary} \label{coro:mean_cov_prop}
The propagation of mean and covariance matrix in (\ref{eq:exact_EOM}) can be organized into a parametrized form:
\begin{equation} \label{eq:mean_cov_prop_exact}
\left \{
\begin{aligned}
    (\dot{\mu}\mu^{-1})^{\vee}\!&\!=\!\langle J_r^{-1}\rangle^{-1}\left\langle \frac{1}{2} \frac{\partial J_l^{-1}}{\partial x_k}(HH^T J_l^{-T})\boldsymbol{e}_k +J_l^{-1}{\boldsymbol{h}^c}\right\rangle\\
 \dot{\Sigma}\qquad \!&\!=\!\bigg \langle\text{sym}\bigg[ \big(\frac{1}{2}\frac{\partial J_l^{-1}}{\partial x_k}(HH^T J_l^{-T})\boldsymbol{e}_k \!-\! J_r^{-1} (\dot{\mu}\mu^{-1})^{\vee}\\
&\qquad  +J_l^{-1}{\boldsymbol{h}^c} \big)\boldsymbol{x}^T  \bigg]+J_l^{-1}HH^TJ_l^{-T}  \bigg \rangle
\end{aligned}
\right .
\end{equation}
\begin{proof}
The derivation is in Appendix \ref{proof:coro}
.\end{proof}
\end{corollary}
In this context, the averaging $\langle \cdot \rangle$ is
\begin{equation} \label{eq:ave_rho}
    \langle \varphi(k) \rangle = \int_{D\subseteq R^{N}} \varphi(k(\boldsymbol{x}))\rho(k(\boldsymbol{x}),t)|\det J_l(\boldsymbol{x})|d\boldsymbol{x},
\end{equation}
where $D$ is the domain of the exponential map, \textit{e}. \textit{g}. $\lVert\boldsymbol{x}\rVert_2<\pi$ for $SO(3)$.
In most cases, the integration does not have a closed-form expression.
We can approximate it by two methods: numerical quadrature and moment expansion.
For the first approach, a Gaussian distribution assumption is often necessary.

\subsubsection{Gaussian solution to Fokker-Planck equation}
Assuming we have obtained the mean and covariance during $0\leq t\leq T$, we can construct an approximate solution to (\ref{eq:leftFPE_EOM}) by defining a Gaussian distribution in exponential coordinates
\begin{equation} \label{eq:concentrated_gaussian}
\begin{aligned}
    \tilde{\rho}(\boldsymbol{x},t)\doteq \frac{1}{Z}\cdot
 \exp \left( -\frac{1}{2} 
 \boldsymbol{x}^T\Sigma^{-1}(t)\boldsymbol{x}\right)
 \end{aligned}
\end{equation}
where $\boldsymbol{x}=\log^{\vee}(g\circ \mu^{-1}(t))$,
and convert it to a distribution on the group by using
\begin{equation}
\begin{aligned}
f(g\circ \mu^{-1}(t),t)dg
&\doteq f(\exp(\boldsymbol{x})\circ \mu^{-1}(t),t)|\det J_l(\boldsymbol{x})|d\boldsymbol{x}\\
&=\rho\left(\exp (\boldsymbol{x}),t\right)|\det J_l(\boldsymbol{x})|d\boldsymbol{x}\\
&= \tilde{\rho}(\boldsymbol{x},t)d\boldsymbol{x},
\end{aligned}
\end{equation}
where $\rho(k,t)=f(k\circ \mu(t),t)$ is defined in the re-centering step.
For small covariance, the normalizing constant 
$Z\approx (2\pi)^{N/2}|\det \Sigma(t)|^{1/2}$ and this Gaussian distribution is almost the same as the concentrated Gaussian in \cite{bourmaud2015continuous,long2013banana,sjoberg2021lie,barfoot2014associating}:
\begin{equation}
    g=\exp{(\boldsymbol{x})}\circ \mu(t),
\end{equation}
where $\boldsymbol{x}\sim \mathcal{N}(\boldsymbol{0},\Sigma(t))$.
It is not hard to see that the mean of the Gaussian distribution is $\mu(t)$.
For a short time, when the probability density is concentrated on the exponential coordinate, the right group-theoretic covariance of the Gaussian (\ref{eq:concentrated_gaussian}) is $\Sigma(t)$:
\begin{equation*}
\begin{aligned}
    &\int_G \left[\log^\vee\left(g \circ \mu^{-1}(t)\right)\right] \left[\log^\vee\left(g \circ \mu^{-1}(t)\right)\right]^T f(g,t)dg\\
    =& \int_{D\subseteq \mathbb{R}^N}\boldsymbol{x}\boldsymbol{x}^T \tilde{\rho}(\boldsymbol{x},t)d\boldsymbol{x}\\
    \approx & \,\,\Sigma(t),
\end{aligned}
\end{equation*}
where $D$ is the domain of the exponential coordinate.

By using this Gaussian assumption, the integration in (\ref{eq:ave_rho}) becomes
\begin{equation}
    \begin{aligned}
            \langle \varphi(k) \rangle &\approx
            \int_{\mathbb{R}^N}\varphi(k(\boldsymbol{x}))\tilde{\rho}(\boldsymbol{x},t)d\boldsymbol{x},
    \end{aligned}
\end{equation}
where $\tilde{\rho}(\boldsymbol{x},t)$ defined in (\ref{eq:concentrated_gaussian}) is a zero-mean Gaussian.
We proceed to use it to derive an approximate propagation formula for (\ref{eq:mean_cov_prop_exact}) by numerical quadrature.

\subsubsection{Approximate propagation by quadrature} \label{sec:app_sol_qua}
The unscented transform is widely used for calculating the expectation of a function when the probability distribution is approximated by a Gaussian.
If the ground true distribution is indeed a Gaussian, the unscented transform is exact to the second-order.
In this case, higher-order quadrature methods can also be used, like the fourth-order conjugate unscented transform \cite{adurthi2018conjugate} and Gauss-Hermite quadrature.
The ``$k$th-order'' here means when the integrand is a $k$th-order polynomial, the numerical quadrature gives an exact value.

For all quadrature methods, the key idea is that a set of ``sigma points'' can be computed from the mean and covariance, and the integration is approximated by a linear sum of the function evaluated at these sigma points:
\begin{equation}
\begin{aligned}
    (\boldsymbol{\xi}_i,\, w_i)&= \text{Algorithm}(\mu, \Sigma),\quad i=1,2,...,M,\\
    &\int_{\mathbb{R}^N}\varphi(\boldsymbol{x})d\boldsymbol{x}\approx \sum_{i=1}^M w_i \varphi(\boldsymbol{\xi}_i).
\end{aligned}
\end{equation}
Different algorithms for computing the sigma points $\boldsymbol{\xi}_i$ and weight $w_i$ lead to different approximation precision and speed.

In our case, we have three integrations in the propagation formula (\ref{eq:mean_cov_prop_exact}).
To prevent repeated computation, the integration in the mean propagation should be computed first.
The following intermediate terms can be stored and used for computing covariance propagation:
\begin{equation}
\begin{aligned}
    &J_r^{-1}(\boldsymbol{\xi}_i),\quad  \frac{\partial J_l^{-1}(\boldsymbol{\xi}_i)}{\partial x_k}(HH^T J_l^{-T}(\boldsymbol{\xi}_i))\boldsymbol{e}_k, \\
     &J_l^{-1}(\boldsymbol{\xi}_i)H, \quad J_l^{-1}(\boldsymbol{\xi}_i){\boldsymbol{h}^c}(\boldsymbol{\xi}_i).
\end{aligned}
\end{equation}
We only need to calculate the linear sum of these stored quantities multiplying sigma points $\boldsymbol{\xi}_i$ or $(\dot{\mu}\mu^{-1})^{\vee}$ for propagating covariance.

\subsubsection{Approximate propagation by expansion of moments}
Another approximation method for (\ref{eq:mean_cov_prop_exact}) is expansion of moments, which expands the quantities inside ``$\langle \cdot \rangle$'' into polynomials and truncates:
\begin{equation} \label{eq:Eom_example}
    \begin{aligned}
    \langle M(\boldsymbol{x})\rangle
    &= \langle A_0+ A_1^i x_i+ A_2^{ij} x_i x_j+...  \rangle \\
    &\approx A_0 +   A_2^{ij} \langle x_i x_j\rangle, 
    \end{aligned}
\end{equation}
where the first-order term vanishes because of (\ref{eq:rho_mean}). 
We do not need to assume the form of the distribution for this method.
The approximation error of truncation is discussed in a moment closure literature \cite{kuehn2016moment} as an ``extremely difficult'' problem.
So we will not delve into analyzing the error but simply call (\ref{eq:Eom_example}) a second-order approximation, since it is exact when the integrand is a second-order polynomial. 
The following lemma is used for expansion.
\begin{lemma} \label{jacobian_expansion} \cite{bullo1995proportional}
    When using the exponential coordinate to parameterize the group, the Taylor expansion of $J_{l}^{-1}(\boldsymbol{x})$ and $J_{r}^{-1}(\boldsymbol{x})$ is:
    \begin{equation}
        J_l^{-1}(\boldsymbol{x})= -\frac{1}{2}[ad_X]+\sum_{k=0}^{\infty}\frac{B_{2k}}{(2k)!} [ad_{X}]^{2k},
    \end{equation}
    \begin{equation}
        J_r^{-1}(\boldsymbol{x})= \frac{1}{2}[ad_X]+\sum_{k=0}^{\infty}\frac{B_{2k}}{(2k)!} [ad_{X}]^{2k},
    \end{equation}
    where $B_{n}$ is the $n$th Bernoulli number.
\end{lemma}
From the lemma, we see that the Taylor series have no odd-order term except for the first-order term.
The first four terms in the Taylor series are
\begin{equation}
    J_l^{-1}(\boldsymbol{x})=\mathbb{I}-\frac{1}{2}[ad_X]+\frac{1}{12}[ad_X]^2-\frac{1}{720}[ad_X]^4+\mathcal{O}(||\boldsymbol{x}||^6),
\end{equation}
\begin{equation}
    J_r^{-1}(\boldsymbol{x})=\mathbb{I}+\frac{1}{2}[ad_X]+\frac{1}{12}[ad_X]^2-\frac{1}{720}[ad_X]^4+\mathcal{O}(||\boldsymbol{x}||^6).
\end{equation}
Now we can explain the invertibility of $\langle J_r^{-1} \rangle$ in (\ref{eq:mean_cov_prop_exact}), which is $\langle E^r_i(\boldsymbol{x})\boldsymbol{e}_i^T \rangle$ in (\ref{eq:exact_EOM}).
Since the distribution starts from a Dirac delta function, this matrix is the identity matrix at $t=0$, using the expansion formula.
Assume the smoothness of the solution with respect to time, it will be invertible for short times.

Applying the Taylor expansion to (\ref{eq:mean_cov_prop_exact}), we have a second-order approximate propagation formula for the mean and covariance matrix:
\begin{theorem} \label{thm:second_order_app}
When using the Taylor expansion and neglecting moments higher than the second order, the equation (\ref{eq:mean_cov_prop_exact}) can be approximated as
\begin{equation} \label{eq:second_order_prop}
\left\{
    \begin{aligned}
        (\dot{\mu} \mu^{-1})^{\vee}\!&\approx \boldsymbol{h}(\mu,t)+  (M^{\mu})_{ij}\Sigma_{ij }\\
        \dot{\Sigma}\quad&\approx HH^T+(M^{\Sigma})_{ij}\Sigma_{ij}
    \end{aligned},
    \right .
\end{equation}
where the coefficients are
\begin{equation} \label{eq:second_order_coef}
\left \{
    \begin{aligned}
        (M^{\mu})_{ij}&=-\frac{1}{48}(ad_k HH^T ad_j^T ad_i^T +ad_i ad_k HH^T ad_j^T )\boldsymbol{e}_k\\
        &\quad +\frac{1}{2}\frac{\partial^2  {\boldsymbol{h}^c}}{\partial x_i \partial x_j}-\frac{1}{2}ad_i \frac{\partial {\boldsymbol{h}^c}}{\partial x_j},\\
        (M^{\Sigma})_{ij}&=\text{sym}\bigg\{\bigg[ \frac{1}{8}ad_k HH^T ad_i^T \boldsymbol{e}_k+\frac{1}{24}ad_k ad_i HH^T \boldsymbol{e}_k\\
        &\quad -\!\frac{1}{2}ad_i (\boldsymbol{h}+(\dot{\mu}\mu^{-1})^{\vee}) 
        +\frac{\partial {\boldsymbol{h}^c}}{\partial x_i} \!\bigg ]
        \boldsymbol{e}_j^T \\
        &\quad+\frac{1}{12}ad_i ad_j HH^T \bigg\}  +\frac{1}{4}ad_i HH^T ad_j^T
    \end{aligned}
    \right.
\end{equation}
where $ad_i\doteq [ad_{E_i}]$ and $\frac{\partial {\boldsymbol{h}}^c}{\partial x_i}=\frac{\partial \boldsymbol{h}(\exp(\boldsymbol{x}^{\wedge})\circ \mu,t)}{\partial x_i}\big |_{\boldsymbol{x}=\boldsymbol{0}}$.
\begin{proof}
The derivation is in the Appendix \ref{proof:2nd_prop}.
\end{proof}
\end{theorem}
Note that $(M^{\mu})_{ij}$ is a vector and $ (M^{\Sigma})_{ij}$ is a matrix, not the entry of a matrix.
Since most of the terms in coefficient matrices are constant or only depend on time, they can be pre-computed and stored.

When truncating the RHS of (\ref{eq:exact_EOM}) at a higher order moment, we should keep the propagation equations to at least the same order to close the equations.
When the parametric form of the distribution is known, higher order moments may be calculated from lower order moments, like the Gaussian distribution, where all moments can be calculated from the mean and covariance \cite{kuehn2016moment}.
In this case, we can use fewer ODEs to achieve a higher-order approximation.


\subsection{Relationship between solutions to left and right Fokker-Planck equations} \label{sec:FPE_connection}
Suppose we have a left Fokker-Planck equation on a matrix Lie group $G$
\begin{equation} \label{eq:l_FPE}
    \frac{\partial f}{\partial t} =E^l_i\left(h_i f\right) + \frac{1}{2}E^l_iE^l_j(H_{ik}H^T_{kj}f),
\end{equation}
and another right Fokker-Planck equation
\begin{equation} \label{eq:r_FPE}
    \frac{\partial f^*}{\partial t} =-E^r_i\left(h^*_i f^*\right) + \frac{1}{2}E^r_iE^r_j(H^*_{ik}{H^*_{kj}}^Tf^*),
\end{equation}
where $h_i^*(g,t)=-h_i(g^{-1},t)$ and $H_{ij}^*(g,t)=H_{ij}(g^{-1},t)$.
The following theorem states the relationship between the solutions of these two equations:
\begin{theorem} \label{thm:l_f_connection}
    If $f(g,t)$ is a solution to (\ref{eq:l_FPE}), then $f^*(g,t)\doteq f(g^{-1},t)$ is a solution to (\ref{eq:r_FPE}).
    \begin{proof}
    The proof is in Appendix \ref{proof:lf_connection}.
    \end{proof}
\end{theorem}
The mean and covariance of these two distributions are also related by the following theorem:
\begin{theorem} \label{theorem:left_right_statistics}
    Denote the right group-theoretic mean of $f(g,t)$ as $\mu_r(t)$ and $\Sigma_r(t)$, the left group-theoretic mean of $f^*(g,t)\doteq f(g^{-1},t)$ as $\mu_l^*(t)$ and $\Sigma_l^*(t)$.
    These two means and covariances are related by
    \begin{equation}
        \mu_l^*(t)=[\mu_r(t)]^{-1},
    \end{equation}
    \begin{equation}
        \Sigma^*_l(t)=\Sigma_r(t).
    \end{equation}
\end{theorem}
\begin{proof}
The proof is in Appendix \ref{proof:mean_relation}.
\end{proof}
Then, if we have a right FPE, we can convert it into a left FPE where we have derived a propagation formula for the right mean and covariance and use the following relationship:
\begin{equation}({\mu_l^*}^{-1}\dot{\mu}_l^*)^{\vee}=-(\dot{\mu}_r \mu_r^{-1})^{\vee},
\end{equation}
\begin{equation}
    \dot{\Sigma}_l^*=\dot{\Sigma}_r,
\end{equation}
to derive the propagation of the left mean and covariance.

\subsection{Relathionship to Kalman filters} \label{sec:EKF}
The propagation steps of Kalman filters on Euclidean space can be derived from the previous propagation formula as well.
To do that, we need to represent an element of $\mathbb{R}^N$ as a matrix Lie group.
The trick for the direct product group in Sec. \ref{sec:back_A} applies here
\begin{equation}
    g(\boldsymbol{x})\doteq \text{diag}(e^{\boldsymbol{x}})=
  \begin{pmatrix}
    e^{x_1} & & \\
    & \ddots & \\
    & & e^{x_N}
  \end{pmatrix}.
\end{equation}
It is obvious that $g(\boldsymbol{x}_1)g(\boldsymbol{x}_2)=g(\boldsymbol{x}_1+\boldsymbol{x}_2)$.
Simple calculations show that both the left and the right Jacobians of this group are an identity matrix.
The basis of Lie algebra and Lie derivatives are simple as well
\begin{equation}
    \begin{aligned}
    &E_i=\text{diag}(\boldsymbol{e}_i)\,\,\,\,\, \text{and} \, 
    &E_i^r(f)=-E_i^l(f)=\frac{\partial f}{\partial x_i}.
    \end{aligned}
\end{equation}
Then, we can formulate the SDE on $\mathbb{R}^N$
\begin{equation} \label{eq:SDE_RN}
d\boldsymbol{x}=\boldsymbol{h}(\boldsymbol{x},t)dt+H(t) \;\circledS\; dW,
\end{equation}
as a SDE on $\mathbb{R}^{N\times N}$
\begin{equation}
    \begin{aligned}
d\boldsymbol{x}&=({dg}g^{-1})^{\vee}
    = \boldsymbol{h}(g(\boldsymbol{x}),t)dt+H(t) \;\circledS\; dW.\\
    \end{aligned}
\end{equation}

Denote the mean of $\boldsymbol{x}$ as $\bar{\boldsymbol{x}}$, the covariance matrix as $\Sigma$, and define $\boldsymbol{\mathrm{x}}=\boldsymbol{x}-\bar{\boldsymbol{x}}$.
Using Corollary \ref{coro:mean_cov_prop} and the Jacobians, we have the propagation equations (\ref{eq:mean_cov_prop_exact}) to be
\begin{equation} \label{eq:ekf_exact}
    \begin{cases}
        \dot{\bar{\boldsymbol{x}}}=\mathbb{I}^{-1}\langle \boldsymbol{0} +\boldsymbol{h}(\bar{\boldsymbol{x}}+\boldsymbol{\mathrm{x}},t)\rangle=\langle \boldsymbol{h}(\bar{\boldsymbol{x}}+\boldsymbol{\mathrm{x}})\rangle\\
        \dot{\Sigma}=HH^T+ \langle \boldsymbol{h}(\bar{\boldsymbol{x}}+\boldsymbol{\mathrm{x}},t)\boldsymbol{\mathrm{x}}^T +\boldsymbol{\mathrm{x}} {\boldsymbol{h}^T(\bar{\boldsymbol{x}}+\boldsymbol{\mathrm{x}},t)}\rangle
    \end{cases},
\end{equation}
where $\bar{\boldsymbol{x}}$ is a deterministic variable and $\boldsymbol{\mathrm{x}}\sim \mathcal{N}(\boldsymbol{0},\Sigma)$.
When using the unscented transform to approximate the right-hand side averaging, the resulting formula is equivalent to the propagation step of a continuous-time unscented Kalman filter (UKF) in \cite{sarkka2007unscented} by the fact that:
\begin{equation}
    \langle \, \boldsymbol{h}(\bar{\boldsymbol{x}}+\boldsymbol{\mathrm{x}},t)\boldsymbol{\mathrm{x}}^T\rangle = \langle \,[\boldsymbol{h}(\bar{\boldsymbol{x}}+\boldsymbol{\mathrm{x}},t)-\boldsymbol{h}(\bar{\boldsymbol{x}},t)]\boldsymbol{\mathrm{x}}^T\rangle.
\end{equation}
When expanding $\boldsymbol{h}(\bar{\boldsymbol{x}}+\boldsymbol{\mathrm{x}},t)$ to the first order, we have
\begin{equation}
\left\{
\begin{aligned}
        \dot{\bar{\boldsymbol{x}}}&=\langle \boldsymbol{h}(\bar{\boldsymbol{x}},t) + O(||\boldsymbol{\mathrm{x}}||^2) \rangle \\ 
        &\approx \boldsymbol{h}(\bar{\boldsymbol{x}},t)\\
        \dot{\Sigma}&=\langle HH^T\!+ \frac{\partial \boldsymbol{h}(\bar{\boldsymbol{x}},t)}{\partial \boldsymbol{\mathrm{x}}^T}\boldsymbol{\mathrm{x}}\boldsymbol{\mathrm{x}}^T \!+\boldsymbol{\mathrm{x}} (\frac{\partial \boldsymbol{h}(\bar{\boldsymbol{x}},t)}{\partial \boldsymbol{\mathrm{x}}^T}\boldsymbol{\mathrm{x}})^T \!+ O(||\boldsymbol{\mathrm{x}}||^3) \rangle\\
        &\approx HH^T+\frac{\partial \boldsymbol{h}(\bar{\boldsymbol{x}},t)}{\partial \boldsymbol{\mathrm{x}}^T}\Sigma + \Sigma (\frac{\partial \boldsymbol{h}(\bar{\boldsymbol{x}},t)}{\partial \boldsymbol{\mathrm{x}}^T})^T
\end{aligned}
\right.
\end{equation}
which is the propagation step of the extended Kalman filter (EKF). 
In fact, since the propagation equation for covariance is given, the propagation of the mean can also include second-order terms.
But in traditional EKF, they are not taken into account.

\section{Propagating Uncertainty for Orientational Dynamics of Rigid Bodies}\label{sec:GoverningEquations}

This section applies the previous theory to propagate the uncertainty of the attitude and angular momentum of a rotating rigid body described in phase space.
Notations and the stochastic differential equation (SDE) in phase space are introduced in Section \ref{sec:dynamics_eq_review}.
We construct a group structure to the phase space by imbuing a direct product group structure in Section \ref{sec:DPG_review}.
The SDE in the phase space can be considered as an SDE on a group.
The corresponding Fokker-Planck equations and propagation formulas are derived.
In Section \ref{sec:UKF_LA}, we briefly review the Lie algebraic unscented Kalman filter \cite{sjoberg2021lie} which we implement as a baseline method.
Note that attitude and angular velocity estimation problem has a long history dating back to the last century. 
There are many existing works \cite{fujikawa1995spacecraft,axelrad1996spacecraft,psiaki2009generalized,humphreys2005magnetometer,sanyal2008global,ma2014magnetometer,srivastava2021attitude} for this problem.
It is also used as a benchmark problem to test a general filtering algorithm in \cite{phogat2020invariant}.

\subsection{Rotational dynamics in phase space}\label{sec:dynamics_eq_review}
Suppose we have a three-dimensional rigid body, whose inertia matrix is $I$ in a body-fixed frame.
Its altitude is represented by a $3\times 3$ rotation matrix $R\in SO(3)$. 
The angular velocity in the body-fixed frame is calculated as $\bm{\omega} = (R^T\dot{R})^\vee$, where the $\vee$ operation is for the Lie algebra of $SO(3)$.
Then angular momentum is $\bm{\ell} = I\bm{\omega}$.
We use Hamiltonian mechanics to describe the motion of the rigid body.
The phase space is the space of angular momenta $(\mathbb{R}^3,+)$ and the space of rotation $SO(3)$.

Assume the rigid body is rotating in a viscous fluid subject to a deterministic torque $\boldsymbol{N}^*(t)$, a linear viscous torque $-C\boldsymbol{\omega}$, and a random disturbance torque $\boldsymbol{\eta}$ modeled by Brownian noise.
The classical Euler's equation describes the evolution of angular momentum:
\begin{equation*}
    \dot{\boldsymbol{\ell}} + (I^{-1}\boldsymbol{\ell})\times\boldsymbol{\ell} = -CI^{-1}\boldsymbol{\ell} + \boldsymbol{N}^* + \bm{\eta}.
\end{equation*}
By expressing the noise term explicitly, $\bm{\eta} dt = B\, \circledS \,d\bm{W}$, where $d\bm{W}$ is a Wiener process increment for a short time $dt$, we can write the Euler's equation as a stochastic differential equation (SDE)
\begin{equation}\label{eq:AngMomentSDE}
    d\boldsymbol{\ell} + (I^{-1}\boldsymbol{\ell})\times\boldsymbol{\ell}\;dt = (-CI^{-1}\boldsymbol{\ell}\; + \boldsymbol{N}^*)\,dt+B\;\circledS\;d\bm{W}.
\end{equation}
We assume the noise coefficient matrix $B$ only depends on time. 
The evolution of altitude is determined by
\begin{equation}\label{eq:RotSDE}
    (R^T\dot{R})^\vee dt = I^{-1}\bm{\ell}\;dt.
\end{equation}
Since $(\mathbb{R}^3,+)$ is an Abelian group, we can imbue a direct product group structure to the phase space and reformulate the dynamics equation to be on the group.

\subsection{Error propagation on $SO(3)\times\mathbb{R}^3$} \label{sec:DPG_review}
We first discuss the direct product group.
The matrix representation of an element of $SO(3)\times\mathbb{R}^3$ is
\begin{equation}
    g(R,\boldsymbol{\ell})=\begin{pmatrix}
        R & \mathbb{O} \\
        \mathbb{O} & \text{diag}(e^{\boldsymbol{\ell}}),
    \end{pmatrix}
\end{equation}
where $\text{diag}(e^{\boldsymbol{\ell}})$ is a diagonal matrix, whose $i$th element is $e^{\ell_i}$.
The basis of the Lie algebra of $SO(3)\times\mathbb{R}^3$ is:
$$
E_i = \begin{cases}
\begin{pmatrix}
    {E}_i^{\scriptscriptstyle {SO(3)}} & \mathbb{O}\\
    \mathbb{O} & \mathbb{O}
\end{pmatrix}, \,\, i=1,2,3 \\
\begin{pmatrix}
    \mathbb{O} & \mathbb{O}\\
    \mathbb{O} & \text{diag}(\boldsymbol{e}_{i-3})
\end{pmatrix}, \, \, i=4,5,6
\end{cases},
$$
where ${E}_i^{\scriptscriptstyle {SO(3)}}$ is the basis of $SO(3)$ and $\text{diag}(\boldsymbol{e}_i)$ is a $3\times 3$ diagonal matrix whose $i$th diagonal element is $1$ and other elements are $0$.
If $X$ is in the Lie algebra of $SO(3)\times \mathbb{R}^3$
\begin{equation}
    X=\begin{pmatrix}
        \Omega & \mathbb{O}\\
        \mathbb{O} & \text{diag}(\boldsymbol{\ell})
    \end{pmatrix},
\end{equation}
then
\begin{equation} \label{eq:ad_dir}
    [ad(X)]=\begin{pmatrix}
        \Omega & \mathbb{O}\\
        \mathbb{O} & \mathbb{O}
    \end{pmatrix}.
\end{equation}
The right and left Jacobians, $\mathcal{J}_l$ and $\mathcal{J}_r$, are
\begin{equation} \label{eq:jacobian_dir}
    \mathcal{J}_r = 
    \begin{pmatrix}
    J_r^{\scriptscriptstyle{SO(3)}}(R) && \mathbb{O}\\
    \mathbb{O} && \mathbb{I}
    \end{pmatrix}
    \;\;\text{and}\;\;
    \mathcal{J}_l = 
    \begin{pmatrix}
    J_l^{\scriptscriptstyle{SO(3)}}(R) && \mathbb{O}\\
    \mathbb{O} && \mathbb{I}
    \end{pmatrix},
\end{equation}
where $J_r^{\scriptscriptstyle{SO(3)}}$ and $J_l^{\scriptscriptstyle{SO(3)}}$ are the Jacobians for $SO(3)$.
It is easy to see that this group is also unimodular.
Using the following relationship
\begin{equation}
    (g^{-1}\dot{g})^{\vee}=\begin{pmatrix}
        (R^T \dot{R})^{\vee} \\
        \dot{\boldsymbol{\ell}}
        \end{pmatrix}
        =\begin{pmatrix}
            \boldsymbol{\omega} \\
            \dot{\boldsymbol{\ell}}
        \end{pmatrix},
\end{equation}
we have the dynamics equation
\begin{equation} \label{eq:dirSDE}
\begin{aligned}
    (g^{-1}\dot{g})^{\vee}dt=\begin{pmatrix}
        I^{-1}\boldsymbol{\ell} \\
        \boldsymbol{\ell}\times (I^{-1}\boldsymbol{\ell})-CI^{-1}\boldsymbol{\ell}+\boldsymbol{N}^*
    \end{pmatrix}dt +\begin{pmatrix}
    \mathbb{O} && \mathbb{O}\\
    \mathbb{O} && B \end{pmatrix}\circledS\;d\boldsymbol{W}.
    \end{aligned}
\end{equation}
Calculations lead to the following result:
\begin{prop}
The right Fokker-Planck equation corresponding to the SDE (\ref{eq:dirSDE}) is
\begin{equation}\label{eq:FPE_Dir_Group}
    \frac{\partial f}{\partial t} =  -E^r_i(h_i^* f) + \frac{1}{2} (B'B'^T)_{ij}E^r_iE^r_j f,
\end{equation} where
\begin{equation} \label{eq:h_dir}
    \boldsymbol{h}^*(g(R,\boldsymbol{\ell}),t) = \begin{pmatrix}
        I^{-1}\boldsymbol{\ell} \\
        \boldsymbol{\ell}\times (I^{-1}\boldsymbol{\ell})-CI^{-1}\boldsymbol{\ell}+\boldsymbol{N}^*
    \end{pmatrix},
\end{equation}
and
\begin{equation} \label{def:B_prime}
B' = 
\begin{pmatrix}
\mathbb{O} && \mathbb{O}\\
\mathbb{O} && B
\end{pmatrix}.
\end{equation}
\end{prop}

Since the propagation formula (\ref{eq:mean_cov_prop_exact}) is developed for left FPEs, we utilize Theorem \ref{thm:l_f_connection} to convert (\ref{eq:FPE_Dir_Group}) into a left FPE by defining
\begin{equation}
    \boldsymbol{h}(g,t) \doteq -\boldsymbol{h}^*(g^{-1},t)= \begin{pmatrix}
        I^{-1}\boldsymbol{\ell} \\
        -\boldsymbol{\ell}\times (I^{-1}\boldsymbol{\ell})-CI^{-1}\boldsymbol{\ell}-\boldsymbol{N}^*
    \end{pmatrix}.
\end{equation}
Because of the special structure of the Jacobian of the direct product group (\ref{eq:jacobian_dir}) and the noise matrix (\ref{def:B_prime}), the propagation equation (\ref{eq:mean_cov_prop_exact}) for the left FPE can be simplified to be
\begin{equation} \label{eq:exact_dir_l}
\left\{
\begin{aligned}
    (\dot{\mu}_r{\mu}^{-1}_r)^{\vee}\!&\!=\langle J_r^{-1}\rangle^{-1}\langle J_l^{-1}{\boldsymbol{h}^c_r}\rangle\\
\quad \dot{\Sigma}_r \qquad\!&\!=\!\bigg \langle\text{sym}\big[ (- J_r^{-1} (\dot{\mu}_r{\mu}^{-1}_r)^{\vee} +J_l^{-1}{\boldsymbol{h}^c_r} \big)\boldsymbol{x}^T  \big]+B'B'^T  \bigg \rangle
\end{aligned}
\right.
\end{equation}
where $\boldsymbol{h}^c_r(k,t)\doteq\boldsymbol{h}(k\circ \mu_r,t)$.
Using Theorem \ref{theorem:left_right_statistics}, we have the reformulated propagation equation for the right FPE (\ref{eq:FPE_Dir_Group}):
\begin{equation} \label{eq:exact_dir_r}
\left\{
\begin{aligned}
    ({\mu^*_l}^{-1}\dot{\mu}_l^*)^{\vee}\!&\!=-\langle J_r^{-1}\rangle^{-1}\langle J_l^{-1}{\boldsymbol{h}^c_l}\rangle\\
\quad \dot{\Sigma}_l^* \qquad\!&\!=\!\bigg \langle\text{sym}\big[ ( J_r^{-1} ({\mu^*_l}^{-1}\dot{\mu}_l^*)^{\vee} +J_l^{-1}{\boldsymbol{h}^c_l} \big)\boldsymbol{x}^T  \big]+B'B'^T  \bigg \rangle
\end{aligned}
\right.
\end{equation}
where $\boldsymbol{h}^c_l(k,t)\doteq\boldsymbol{h}(k\circ {\mu_l^*}^{-1},t)$.

Although writing the elements of the direct product group as a matrix yields compact equations (\ref{eq:exact_dir_r}), many terms in the equation turn out to be zeros, which wastes memory and computational resources.
So we present the propagation formula for rotation and angular momentum separately.
We denote the rotation and angular momentum part of the mean $\mu_l^*$ as $\bar{R}$ and $\bar{\boldsymbol{\ell}}$.
The logarithm $\boldsymbol{x}=\log^{\vee}k$ is separated into $\boldsymbol{x}_R=\boldsymbol{x}_{1:3}$ and $\boldsymbol{x}_{{\ell}}=\boldsymbol{x}_{4:6}$.
The covariance matrix is decomposed as $\Sigma_{RR}={(\Sigma_l^*)}_{1:3,1:3}$, $\Sigma_{\ell \ell}={(\Sigma_l^*)}_{4:6,4:6}$, and $\Sigma_{R{\ell}}={(\Sigma_l^*)}_{1:3,4:6}$.
In what follows, all the Lie group operations and Jacobians are for $SO(3)$.
The expansions of Jacobians of $SO(3)$ in (\ref{eq:left_Jacobian_SO3}, \ref{eq:right_Jacobian_SO3}) are
\begin{equation} \label{eq:SO3_rJacobian_approx}
J_r^{-1}(\boldsymbol{x}_R) 
=I+\frac{1}{2}X_R+\frac{1}{12}X^2_R+\mathcal{O}(\lVert\boldsymbol{x}_R\rVert^4),
\end{equation}
\begin{equation} \label{eq:SO3_lJacobian_approx}
J_l^{-1}(\boldsymbol{x}_R) 
=I-\frac{1}{2}X_R+\frac{1}{12}X^2_R+\mathcal{O}(\lVert\boldsymbol{x}_R\rVert^4),
\end{equation}
which can be derived by expanding (\ref{eq:left_Jacobian_SO3}, \ref{eq:right_Jacobian_SO3}) into Taylor series or employing Lemma \ref{jacobian_expansion}.
The explicit expression for ${\boldsymbol{h}}^c_l(k(\boldsymbol{x}),t)$ is
\begin{equation} \label{eq:tilde_h_dir}
{\boldsymbol{h}}^c_l(k(\boldsymbol{x}),t)=
    \begin{pmatrix}
        I^{-1}({\boldsymbol{x}}_{\ell}-\bar{\boldsymbol{\ell}})\\
        -(C+(\boldsymbol{x}_{\ell}-\bar{\boldsymbol{\ell}})^{\wedge})I^{-1}(\boldsymbol{x}_{\ell}-\bar{\boldsymbol{\ell}})\!-\!\boldsymbol{N}^*
    \end{pmatrix},
\end{equation}
where we use ``$\wedge$'' of $SO(3)$ to convert the cross product in (\ref{eq:h_dir}) into matrix multiplication.

We first introduce the approximate propagation method based on a second-order numerical quadrature. 
We approximate right-hand side averaging in (\ref{eq:exact_dir_r}) by unscented transform \cite{julier1995new,sarkka2023bayesian} and call it ``UTD'', Unscented Transform for the Direct product group.
It approximates the averaging of all functions in (\ref{eq:exact_dir_r}) by a weighted sum of $(2n+1)$ sigma points
\begin{equation}
    \langle \varphi \rangle \doteq \int_{D\subseteq R^n} \varphi(\boldsymbol{x})\tilde{\rho}(\boldsymbol{x},t)d\boldsymbol{x}\approx \sum_{i=-n}^n w_i \varphi(\boldsymbol{\xi}_i),
\end{equation}
where,
\begin{equation}  w_0=\kappa/(n+\kappa),\qquad\boldsymbol{\xi}_0=\boldsymbol{0}
\end{equation}
and
\begin{equation}
    \begin{aligned}
       &\qquad w_i=1/[2(n+\kappa)],\,\, \\ \boldsymbol{\xi}_i=\text{sgn}&(i)\left(\sqrt{(n+\kappa)\Sigma_l^*}\right)_{\text{abs}(i)},\,\, i=\pm 1, \, \pm 2, ...\,,\pm n.
    \end{aligned}
\end{equation}
The square root above stands for Cholesky decomposition and its subscript is the column index of the matrix.
We use parameters $\kappa=-3$ to satisfy the constraint $\kappa+n=3$.
The technique introduced in Section \ref{sec:app_sol_qua} can be used to speed up computation.

Then we introduce the approximate propagation based on expansion of moments.
Substituting (\ref{eq:SO3_rJacobian_approx}, \ref{eq:SO3_lJacobian_approx}, \ref{eq:tilde_h_dir}) into (\ref{eq:exact_dir_r}) and removing third and higher-order terms, we have the second-order propagation formula of mean
\begin{equation} \label{eq:joint_est_prop}
\left \{
    \begin{aligned}
        (\bar R^T\dot{\bar R})^{\vee}&\approx I^{-1}\bar{\boldsymbol{\ell}}+ \frac{1}{2}(\Sigma_{R\ell})_{ij}E_i^{\scriptscriptstyle{SO(3)}}I^{-1}\boldsymbol{e}_j \\
        \dot{{\bar{\boldsymbol{\ell}} }} \qquad &\approx (\bar{\boldsymbol{\ell}}^{\wedge}-C)I^{-1}\bar{\boldsymbol{\ell}}+\boldsymbol{N}^*+(\Sigma_{\ell\ell})_{ij}E_i^{\scriptscriptstyle{SO(3)}} I^{-1}\boldsymbol{e}_j
    \end{aligned}\right.
\end{equation}
where the terms calculated by covariance are second-order terms.
The propagation equation for covariance is
\begin{equation} \label{eq:4th_prop}
\left \{
\begin{aligned}
    \dot{\Sigma}_{RR} &\approx \text{sym}\left(\Sigma_{R\ell} I^{-1}-\frac{1}{2}\left(\bar R^T\dot{\bar R}+(I^{-1}\bar{\boldsymbol{\ell}})^{\wedge}\right)\Sigma_{RR}\right)\\
    \dot{\Sigma}_{R\ell}&\approx -\frac{1}{2}\big((I^{-1}\bar{\boldsymbol{\ell}})^{\wedge}+\bar R^T\dot{\bar R}\big)\Sigma_{R\ell}+I^{-1}\Sigma_{\ell\ell}+\Sigma_{R\ell}(I^{-1}\bar{\boldsymbol{\ell}})^{\wedge}\\&\qquad-\Sigma_{R\ell} I^{-1}(C+\bar{\boldsymbol{\ell}}^{\wedge})\\
    \dot{\Sigma}_{\ell \ell}&\approx BB^T+\text{sym}\left(\left[ (\bar{\boldsymbol{\ell}}^{\wedge}-C)I^{-1}-(I^{-1}\bar{\boldsymbol{\ell}})^{\wedge} \right]\Sigma_{\ell\ell}\right)
    \end{aligned}
    \right .
\end{equation}
We refer to this formula as ``EMD2'', Expansion of Moments on Direct product group to the 2nd-order.
If we do not include the second-order terms in calculating the mean, the simplified formula is the mean propagation used in \cite{barrau2016invariant,brossard2017unscented,bourmaud2015continuous}, which we call ``EMD0''.

\subsection{Lie algebraic unscented Kalman filter} \label{sec:UKF_LA}
We implement the propagation step of the Lie algebraic unscented Kalman filter \cite{sjoberg2021lie} on the direct product group $SO(3)\times \mathbb{R}^3$ as a baseline method.
The key idea is that a non-zero mean Gaussian distribution on the exponential coordinate can be approximated to the first order by a concentrated Gaussian distribution.
That is, if we have a random variable $Y=\exp([\bar{\boldsymbol{\xi}}+\boldsymbol{\xi}]^{\wedge})$, where $\boldsymbol{\xi}\sim \mathcal{N}(\boldsymbol{0},\Sigma)$, the random variable $Z=\exp(\bar{\boldsymbol{\xi}}^{\wedge})\circ \exp(\boldsymbol{e}^{\wedge})$, where $\boldsymbol{e}\sim \mathcal{N}\left(\boldsymbol{0},J_r(\bar{\boldsymbol{\xi}})\Sigma J_r^T(\bar{\boldsymbol{\xi}})\right)$, obeys a distribution that approximates the distribution of $Y$ well.
So it is possible to parameterize the transition function by exponential coordinate, use the traditional UKF to propagate uncertainty, construct a non-zero mean Gaussian distribution, and convert it into a concentrated Gaussian distribution.
The detailed sudo-code is in ``Algorithm 1'' in the original paper \cite{sjoberg2021lie}.
Our implementation and choice of hyperparameters are according to the pseudo-code.
We refer to it as ``UKF-LA''.

\section{Numerical Results}\label{sec:Results}

We use the rotational rigid body problem described in Section \ref{sec:dynamics_eq_review} to test our propagation formula.
The equation of motion for the rigid body is Equation (\ref{eq:AngMomentSDE}, \ref{eq:RotSDE}).
We choose a body fixed frame in which the inertia tensor of the rigid body is diagonal: $I = \text{diag}(I_1,I_2,I_3)$. 
We further assume the coefficient matrices for viscosity and noise are both some multiple of an identity matrix: $C=c\mathbb{I}$ and $B=b\mathbb{I}$.

We choose two reference trajectories of angular momentum
\begin{equation} \label{eq:ref_angular}
        \boldsymbol{\ell}^*_1(t) = 
    \begin{pmatrix}
    0\\
    t+1\\
    2t+1
    \end{pmatrix}\,\,\text{and}\,\,\,
    \boldsymbol{\ell}^*_2(t) = 
    \begin{pmatrix}
    1+0.5\sin(2\pi t)\\
    0\\
    0
    \end{pmatrix}
\end{equation}
and numerically evaluate the following equation
\begin{equation}
    \boldsymbol{N}^*= \dot{\boldsymbol{\ell}} + CI^{-1}\boldsymbol{\ell}+ (I^{-1}\boldsymbol{\ell})\times\boldsymbol{\ell} 
\end{equation}
to obtain the deterministic torque.
The time derivative of the angular momentum is approximated by central difference, except for the first and the last step, where we use forward and backward finite difference.
Note that when noise is injected the mean of the angular momentum is not the same thing as the reference trajectory.
We only use the reference to calculate the deterministic torque.

When the inertial tensor in the body-fixed frame is a multiple of the identity matrix, the cross-product term in (\ref{eq:AngMomentSDE}) vanishes.
We try to avoid this special case by choosing the inertial matrix $I$ to be
\begin{equation}
    I = 
    \begin{pmatrix}
    2.070 && 0 && 0\\
    0 && 1.532 && 0\\
    0 && 0 && 1.236
    \end{pmatrix},
\end{equation}
where the relative values are used in \cite{fujikawa1995spacecraft,jayaraman2023lietheoretic}.
To make the effect of viscosity and noise noticeable, we choose $c=1$, $b=1$, and the total time $T=1$.

\subsection{Sampling and data processing}
We simulate the SDEs (\ref{eq:AngMomentSDE}) using a modified improved Euler's scheme \cite{roberts2012modify}.
Since the scheme is proposed for SDEs in Euclidean space, we adapt it for simulating the equation (\ref{eq:RotSDE}) on $SO(3)$, where the arithmetic average does not exist.
At each time step, the adapted update rule for the rotation matrix is
\begin{equation} \label{eq:IE_group}
    R(t+\Delta t)=R(t) \exp\left( \frac{\Delta t}{2} \left(I^{-1}\boldsymbol{\ell}(t)+I^{-1}\boldsymbol{\ell}(t+\Delta t)\right)^{\wedge}\right).
\end{equation}
We totally sample $N_s=5,000,000$ trajectories using $\Delta t=1\times 10^{-3}$.
The $i$th sample at time $t$ is denoted as $\left(\boldsymbol{\ell}_i(t),R_i(t)\right)$.

We proceed to calculate the mean and covariance for samples.
For angular momentum, the sample mean is the arithmetic mean
\begin{equation}
    \bar{\boldsymbol{\ell}}=\frac{1}{N_s}\sum_{i=1}^{N_s} \boldsymbol{\ell}_i.
\end{equation}
For rotation, the right sample mean is calculated by first guessing
\begin{equation} \label{eq:sample_mean_init}
    \bar{R}_s = \exp\left(\frac{1}{N_s}\sum_{i=1}^{N_s} [\log R_i]\right).
\end{equation}
The guess is then refined by
\begin{equation}
    \bar{R}_s = \bar{R}_s \circ \exp\left(\frac{1}{N_s}\sum_{i=1}^{N_s} \log(\bar{R}_s^{-1} \circ  R_i)\right),
\end{equation}
until $||\frac{1}{N_s}\sum_{i=1}^{N_s} \log \left(\bar{R}_s^{-1} \circ  R_i\right)||_F<1\times 10^{-6}$. 
We denote the converged mean as $\bar{R}_s$.
The left sample covariance is then calculated by
\begin{equation}
    \Sigma_s = \frac{1}{N_s}\sum_{i=1}^{N_s}\boldsymbol{x}_i \boldsymbol{x}_i^T
\end{equation}
where $\boldsymbol{x}_i^T=\left[(\boldsymbol{\ell}_i-\bar{\boldsymbol{\ell}})^T,\left(\log^{\vee}(\bar{R}_s^{-1}\circ R_i)\right)^T\right]$.

\subsection{Numerical integration}
We use the improved Euler method to integrate all continuous propagation formulas for high accuracy.
As for UKF-LA, which is proposed for discrete-time setting (equivalent to the forward Euler method), we modify it into an improved Euler scheme by the following procedure. 
At time step $t$, we run the forward Euler equation for two steps to obtain two increments of the mean, $[\Delta t \cdot \tilde{\boldsymbol{\xi}}(t)]$ and $[\Delta t \cdot \tilde{\boldsymbol{\xi}}(t+\Delta t)]$, and two increments of covariance matrices, $[\tilde{\Sigma}(t+\Delta t)-\Sigma(t)]$ and $[\tilde{\Sigma}(t+2\Delta t)-\tilde{\Sigma}(t+\Delta t)]$.
Then we average the increments to propagate the mean and covariance:
\begin{equation*}
\begin{aligned}
    \mu(t+\Delta t)&=\mu(t)\circ \exp\left(\Delta t \cdot \frac{\tilde{\boldsymbol{\xi}}(t)+\tilde{\boldsymbol{\xi}}(t+\Delta t)}{2}\right),\\
    \Sigma(t+\Delta t)&=\frac{1}{2}\left({\Sigma}(t)+\tilde{\Sigma}(t+2\Delta t)\right).
    \end{aligned}
\end{equation*}
For all methods that decompose the covariance matrix by Cholesky decomposition, we set the initial covariance matrix to be a small diagonal matrix $\Sigma\big |_{t=0}=(1\times 10^{-8}) \cdot \mathbb{I}$, which ensures the non-singularity at the initial propagation steps while not influencing the precision.
For other methods, we set $\Sigma\big |_{t=0}=\mathbb{O}$.

\subsection{Evaluation metric}
We evaluate the mean of rotation and angular momentum separately.
For the rotation part, we calculate the difference between the sample and propagated mean and use the Frobenius norm as the error:
\begin{equation}
e_R=\lVert\bar{R}_{s}(t)-\bar{R}_{p}(t)\rVert_F.    
\end{equation}
For angular momentum, we use 2-norm for evaluation:
\begin{equation}
    e_{\ell}=\lVert\bar{\boldsymbol{\ell}}_s(t)-\bar{\boldsymbol{\ell}}_{p}(t)\rVert_2.
\end{equation}
To evaluate the covariance, we also take the difference and calculate the Frobenius norm:
\begin{equation}
    e_{\Sigma}=\lVert\Sigma_s(t)-\Sigma_p(t)\rVert_F.
\end{equation}

\subsection{Experiment results}
In this section, we show the experiment results corresponding to the trajectory $\boldsymbol{\ell}^*_1(t)$ and $\boldsymbol{\ell}^*_2(t)$ in (\ref{eq:ref_angular}).





\begin{figure}[thpb]
      \centering
      \includegraphics[width=0.6\textwidth]{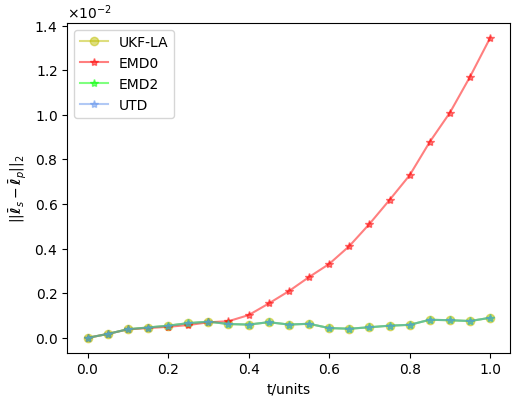}
      \caption{The plot of $e_{\ell}=\lVert\bar{\boldsymbol{\ell}}_s-\bar{\boldsymbol{\ell}}_p\rVert_2$ for Trajectory 1.}
      \label{fig:traj1_ang_error}
\end{figure}
\begin{figure}[thpb]
      \centering
      \includegraphics[width=0.6\textwidth]{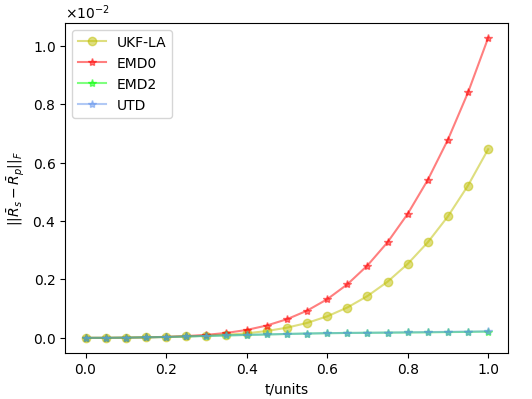}
      \caption{The plot of $e_R=\lVert\bar{R}_s-\bar{R}_p\rVert_F$ for Trajectory 1. }
      \label{fig:traj1_rot_error}
\end{figure}
\begin{figure}[thpb]
      \centering
      \includegraphics[width=0.6\textwidth]{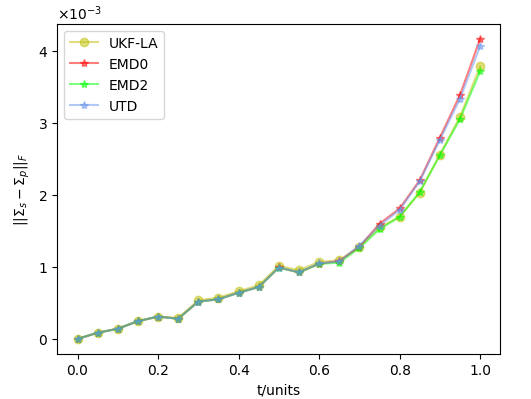}
      \caption{The plot of $e_{\Sigma}=\lVert\Sigma_s-\Sigma_p\rVert_F$ for Trajectory 1.}
      \label{fig:traj1_sigma_error}
\end{figure}

\begin{figure}[thpb]
      \centering
      \includegraphics[width=0.6\textwidth]{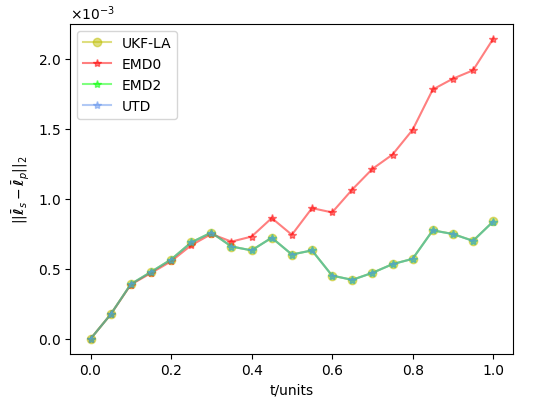}
      \caption{The plot of $e_{\ell}=\lVert\bar{\boldsymbol{\ell}}_s-\bar{\boldsymbol{\ell}}_p\rVert_2$ for Trajectory 2.}
      \label{fig:traj2_ang_error}
\end{figure}
\begin{figure}[thpb]
      \centering
      \includegraphics[width=0.6\textwidth]{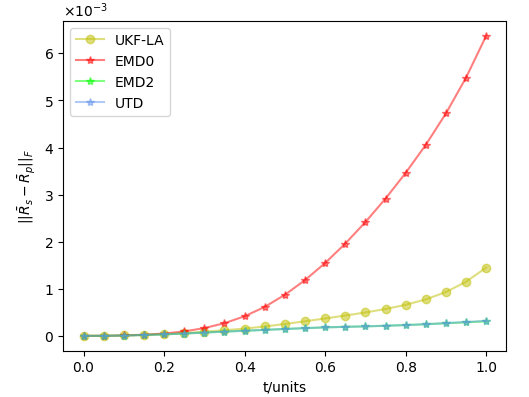}
      \caption{The plot of $e_R=\lVert\bar{R}_s-\bar{R}_p\rVert_F$ for Trajectory 2.}
      \label{fig:traj2_rot_error}
\end{figure}
\begin{figure}[thpb]
      \centering
      \includegraphics[width=0.6\textwidth]{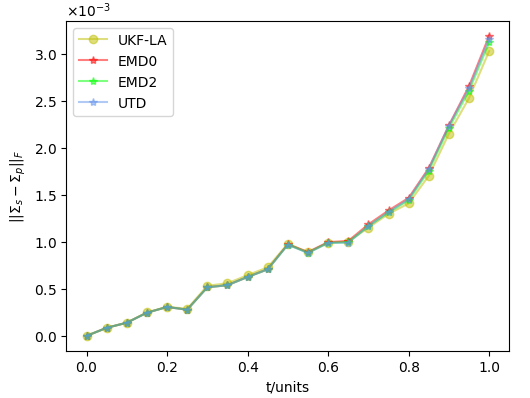}
      \caption{The plot of $e_{\Sigma}=\lVert\Sigma_s-\Sigma_p\rVert_F$ for Trajectory 2.}
      \label{fig:traj2_sigma_error}
\end{figure}

From Figure \ref{fig:traj1_ang_error} and \ref{fig:traj2_ang_error}, for angular momentum, UKF-LA, EMD2, and UTD perform equally well.
When using the direct product group, the transition function of angular momentum is in Euclidean space, not involving a complex group structure, and these three methods are all second-order.
In Figure \ref{fig:traj1_rot_error} and \ref{fig:traj2_rot_error}, for the rotation part, EMD2 and UTD are the best, with minor differences.
The UKF-LA has a relatively larger error that may come from the first-order approximation in mean propagation.
The EMD0, which uses the deterministic transition function to propagate the mean, always has the largest error.
In Figure \ref{fig:traj1_sigma_error} and \ref{fig:traj2_sigma_error}, all methods perform equally well, with UKF-LA being slightly better in Figure \ref{fig:traj2_sigma_error}.

The computation time of different methods is in Table \ref{table:time}.
The EMD2 is $12\%$ slower than EMD0 but is much more accurate in mean estimation.
The UTD and UKF-LA achieve the same accuracy as EMD2 while being 2 to 3 times slower.
All methods are implemented by Python using the NumPy library for linear algebra computation.
All group-related operators, such as the Jacobians and ``$ad$'', are computed using closed-form expressions with special treatment at singularities.

\begin{table}[!h]  \small 
\begin{center} 
 \begin{tabular}{||c c c c c||} 
 \hline 
 Method & UKF-LA & EMD0 & EMD2 & UTD  \\ [0.1ex] 
 \hline\hline
 $T_{traj1}/s$ & 2.17 & 0.66 & 0.74 & 1.77 \\ 
 \hline
$T_{traj2}/s$ & 2.16 & 0.66 & 0.74 & 1.79\\
 \hline
 \end{tabular}
 \caption{Total running time of different methods for trajectory 1 and trajectory 2.} \label{table:time}
\end{center} 
\end{table}

\section{Conclusion}\label{sec:Concl}


This paper proposes a method to propagate uncertainty on unimodular matrix Lie groups that have a surjective exponential map.
A general formula that describes the exact propagation of mean and covariance is derived and approximated by two methods.
The group-theoretic mean propagated by the approximation methods achieves high accuracy with fast computation speed.
The connection between our propagation theory and the propagation step of Kalman filters in Euclidean space is stated.
Then we apply the theory to the attitude and angular momentum error propagation problem using Euler's equation of motion.
We derive a closed-form second-order propagation formula using the expansion of moments and a formula using the unscented transform.
The result shows that these two formulas have better accuracy in estimating the mean of rotation than both baseline methods and the closed-form formula achieves a high computation speed.
For applications that require fast computation, the closed-form formula is more suitable.
When the derivative of the drift term $\boldsymbol{h}(g,t)$ is hard to calculate manually, the quadrature method is recommended.

In the future, the error propagation technique can be extended to propagating quantities for constructing a multimodal distribution solution as what can be done by particle filters \cite{chiuso2000monte,marjanovic2016engineer,zhang2016feedback}. 
The uncertainty propagation method can also be combined with an observation model to be a continuous-discrete filter for various estimation problems, such as the quadrotor UAVs tracking problem \cite{bangura2012nonlinear,lee2010geometric}.


\appendix
\section{APPENDIX}\label{sec:App1}
\subsection{Proof of Theorem \ref{thm:Thm1}} \label{proof:FPE}
The stochastic differential equation in parameters that corresponds to (\ref{eq:GrpStraton}) is
\begin{equation}\label{eq:ParaStraton}
    d\boldsymbol{q} = [J_r(\boldsymbol{q})]^{-1}\tilde{\boldsymbol{h}}(\boldsymbol{q},t)dt + [J_r]^{-1}(\boldsymbol{q})\tilde{H}(\boldsymbol{q},t)\;\circledS\; d\boldsymbol{W}.
\end{equation}
To convert to a corresponding It\^{o} equation, we modify the drift using the correction term
\cite{chirikjian2009stochastic} $\Delta \boldsymbol{h}$ given as
\begin{equation*}\label{eq:ItoCorrection}
 \Delta \tilde{h}_i = \frac{1}{2}\frac{\partial}{\partial q_k}\left([J_r]^{-1}_{ip}\tilde{H}_{pj}\right)[J_r]^{-1}_{kn}\tilde{H}_{nj}.
\end{equation*}
Then
\begin{equation}\label{eq:ParaIto}
    d\boldsymbol{q} = \left([J_r(\boldsymbol{q})]^{-1}\tilde{\boldsymbol{h}}(\boldsymbol{q},t) + \Delta\tilde{\boldsymbol{h}}\right)dt + [J_r(\boldsymbol{q})]^{-1}\tilde{H}(\boldsymbol{q},t) d\boldsymbol{W}
\end{equation}
is the corresponding It\^{o} stochastic differential equation. The Fokker-Planck equation for the probability density function $f(g(\boldsymbol{q}),t) = \tilde{f}(\boldsymbol{q},t)$ corresponding to (\ref{eq:ParaIto}) is (no group-theory is used to derive this, but only the relationship between stochastic differential equations and Fokker-Planck equations for stochastic differential equations evolving on manifolds \cite{chirikjian2009stochastic}):
\begin{equation*}\label{eq:FPEIto}
    \frac{\partial \tilde{f}}{\partial t} +\frac{1}{|J|} \frac{\partial}{\partial q_i}\left(|J|[J_r]^{-1}_{im}\tilde{h}_mf + |J|\Delta \tilde{h}_i\tilde{f}\right) = \frac{1}{2|J|}\frac{\partial^2}{\partial q_i\partial q_j}\left(|J|[J_r]^{-1}_{im}\tilde{H}_{mn}\tilde{H}^T_{np}[J_r]^{-T}_{pj}\tilde{f}\right),
\end{equation*}
where $|J| = |\text{det}\;J_r|$. Note that for a unimodular group, $|J| = |\text{det}\;J_r| = |\text{det}\;J_l|$ where $[J_l(\boldsymbol{q})]$ is the left Jacobian matrix. Additionally, for a unimodular group \cite{chirikjian2011stochastic}, it is also known that
\begin{equation*}\label{eq:detJIdentity}
\frac{\partial}{\partial q_i}(|J|[J_r]^{-1}_{ij}) = \frac{\partial}{\partial q_i}(|J|[J_l]^{-1}_{ij}) = 0,
\end{equation*}
which is equivalent to stating that
\begin{equation*}
    \frac{\partial |J|}{\partial q_i}[J_r]^{-1}_{ij} = -|J|\frac{\partial [J_r]^{-1}_{ij}}{\partial q_i},
\end{equation*}
and likewise for left Jacobians. Using this identity we have
\begin{equation*}
    \frac{\partial \tilde{f}}{\partial t} +  [J_r]^{-T}_{mi}\frac{\partial}{\partial q_i}\left(\tilde{h}_m\tilde{f}\right) + \frac{1}{|J|}\frac{\partial}{\partial q_i}(|J|\Delta \tilde{h}_i\tilde{f}) = \frac{1}{2}\frac{1}{|J|}\frac{\partial^2}{\partial q_i\partial q_j}\left(|J|[J_r]^{-1}_{im}[J_r]^{-T}_{pj}\tilde{H}_{mn}\tilde{H}^T_{np}\tilde{f}\right),
\end{equation*}
and using the expression for the right Lie derivative (\ref{eq:rightDerivativeJT}) we have
\begin{equation*}
\begin{aligned}
    &\frac{\partial f}{\partial t} +  E^r_m\left(h_m f \right) + \frac{1}{|J|}\frac{\partial}{\partial q_i}(|J|\Delta\tilde{h}_i \tilde{f}) - \frac{1}{2}E^r_mE^r_p(H_{mn}H^T_{np}f)\\=&\frac{1}{2}\frac{1}{|J|}\frac{\partial}{\partial q_i}\left[|J|[J_r]^{-1}_{jp}\frac{\partial [J_r]^{-1}_{im}}{\partial q_j}\tilde{H}_{mn}\tilde{H}^T_{np}\tilde{f}\right]\\=&\frac{1}{|J|}\frac{\partial}{\partial q_i}\left[|J|(\Delta \tilde{h}_i -\frac{1}{2} [J_r]^{-1}_{ip}[J_r]^{-1}_{kn}\frac{\partial \tilde{H}_{pj}}{\partial q_k} \tilde{H}_{nj})\tilde{f}\right],
    \end{aligned}
\end{equation*}
which gives
\begin{equation*} 
    \frac{\partial f}{\partial t} +  E^R_m\left(h_m f\right) =\frac{1}{2}E^r_mE^r_p(H_{mn}H^T_{np}\tilde{f}) -\frac{1}{2|J|}\frac{\partial}{\partial q_i}\left[|J| [J_r]^{-1}_{ip}[J_r]^{-1}_{kn}\frac{\partial \tilde{H}_{pj}}{\partial q_k} \tilde{H}_{nj}\tilde{f}\right],
\end{equation*}
and
\begin{equation}\label{eq:rightFPE_correct}\begin{aligned}
    \frac{\partial f}{\partial t} +  E^r_m\left(\left[h_m + \frac{1}{2}\frac{\partial \tilde{H}_{mj}}{\partial q_k} [J_r]^{-1}_{kn}\tilde{H}_{nj}\right]f\right) =&
    \frac{\partial f}{\partial t} +  E^r_m\left(\left[h_m + \frac{1}{2}{H}_{nj}E^r_{n}(H_{mj})\right]f\right) \\=&
    \frac{1}{2}E^r_mE^r_p(H_{mn}H^T_{np}f),
    \end{aligned}
\end{equation}
thereby proving the theorem. 

\subsection{Proof of Theorem \ref{thm:change_var}} \label{proof:change_var}
The following lemma is useful in the proof.
\begin{lemma}
If $\mu(t)$ is an element of a matrix Lie group $G$ and it is infinitely differentiable with respect to $t$, we have
\begin{equation}
\mu(t+\tau)=\exp{\left(\tau\cdot \dot{\mu}\mu^{-1}(t)+O(\tau^2)\right)}\circ \mu(t).
\end{equation}
\begin{proof}
Since $\mu(t+\tau)\mu^{-1}(t)\in G$, the logarithm of it is well-defined. 
We can expand the logarithm as
$$
\log {\left(\mu(t+\tau) \circ \mu^{-1}(t)\right)}=\sum_{i=0}^{\infty} \tau^i A_i(t),
$$
and
$$
\mu(t+\tau)\circ \mu^{-1}(t)=\exp{\left(\sum_{i=0}^{\infty} \tau^i A_i(t)\right)}.
$$
Evaluating both side at $\tau=0$, we have $A_0(t)=\mathbb{O}$.
Taking the derivative with respect to $\tau$ and evaluating at $\tau=0$, we have $A_1(t)=\dot{\mu}\mu^{-1}(t)$.
\end{proof}
\end{lemma}

We proceed to proof the main theorem.
The Fokker-Planck equation (\ref{eq:leftFPE_EOM}) can be rewritten in terms of $\rho(k,t)$ by using the relationship $f(g,t)=\rho(g\circ \mu^{-1}(t),t)$ and making the change of variable $k=g\circ \mu^{-1}$.
The left hand side (LHS) of (\ref{eq:leftFPE_EOM}) becomes
\begin{equation} \label{eq:EOM_change_LHS}
    \begin{aligned}
    \frac{\partial f(g,t)}{\partial t} 
    =& \lim_{\tau \to 0} \frac{ \rho(g\circ \mu^{-1}(t+\tau),t+\tau)-\rho(g\circ \mu^{-1}(t),t)}{\tau} \\
    =& \lim_{\tau \to 0} \frac{1}{\tau}\{ \rho[g\circ \mu^{-1}(t)\circ \exp{(-\tau \cdot \dot{\mu}\mu^{-1})},t] \\ 
    & \,\,\,\, -\rho(g\circ \mu^{-1}(t),t)\}
    + \, \, \, \frac{\partial \rho(g\circ \mu^{-1},t)}{\partial t}\\
    =& \lim_{\tau\to 0} \frac{1}{\tau}[\rho(k \circ \exp(-\tau \! \cdot \! \dot \mu  \mu^{-1}) ,t)\! - \! \rho(k ,t) ] +  \frac{\partial \rho(k,t)}{\partial t}\\
    =& -(\dot \mu \mu^{-1})^{\vee}_i \cdot E^r_i\rho(k,t)  + \frac{\partial \rho(k,t)}{\partial t}.
    \end{aligned}
\end{equation}
For the right hand side (RHS) of (\ref{eq:leftFPE_EOM}), we just take $f(g,t)$ for example:
\begin{equation}
    \begin{aligned}
    E_i^l f (g,t)&\doteq \frac{d}{d\tau} f(e^{-\tau E_i}\circ g,t)\\
    &= \frac{d}{d\tau} f(e^{-\tau E_i}\circ k\circ \mu(t),t)\\
    &= E_i^l f (k\circ \mu(t),t)\\
    &= E_i^l \rho(k,t).
    \end{aligned}
\end{equation}
Note that if the Fokker-Planck equation (\ref{eq:leftFPE_EOM}) is formulated using the right Lie derivative $E^r_i$, the succinct formula above will not hold: an additional adjoint operator will appear.

Define $\boldsymbol{h}^c(k,t)=\boldsymbol{h}(k\circ \mu(t),t)$, we have
\begin{equation} \label{eq:EOM_change_RHS}
    \begin{aligned}
    E_i^l(h_i f)(k\circ \mu(t),t)&=E_i^l({h}_i^c \rho)(k,t),\\
    E_i^l E_j^l (H_{ik}H^T_{kj} f)(k\circ \mu(t),t)&=E_i^l E_j^l ({H}_{ik}{H}^T_{kj} \rho)(k,t).
    \end{aligned}
\end{equation}
Combining (\ref{eq:EOM_change_LHS}) and (\ref{eq:EOM_change_RHS}), we have the formula in the theorem.

\subsection{Proof of Corollary \ref{coro:mean_cov_prop}} \label{proof:coro}
Substituting the formulas in Lemma \ref{lemma:lie_J} into the propagation formula of mean and covariance in Theorem \ref{thm:prop_all_moments}, we have
\begin{equation*}
    \begin{aligned}
        (\dot{\mu}\mu^{-1})^{\vee}&=\langle J_r^{-1}\boldsymbol{e}_i \boldsymbol{e}_i^T \rangle^{-1} \bigg \langle \frac{1}{2}(HH^T)_{ij}(J_l^{-T})_{jk}\frac{\partial J_l^{-1}}{\partial x_k}\boldsymbol{e}_i +{h}^c_i J_l^{-1}\boldsymbol{e}_i \bigg\rangle\\
        &= \langle J_r^{-1}\rangle^{-1} \bigg\langle \frac{1}{2} \frac{\partial J_l^{-1}}{\partial x_k} (HH^T J_l^{-T})_{ik}\boldsymbol{e}_i+J_l^{-1}\boldsymbol{h}^c \bigg \rangle\\
        &=\langle J_r^{-1}\rangle^{-1} \bigg \langle \frac{1}{2}\frac{\partial J_l^{-1}}{\partial x_k}(HH^T J_l^{-T})\boldsymbol{e}_k+J_l^{-1}\boldsymbol{h}^c \bigg \rangle
    \end{aligned}
\end{equation*}
and
\begin{equation*}
    \begin{aligned}
        \dot{\Sigma} &= \! \bigg \langle \frac{1}{2} (HH^T)_{ij} \text{sym} \bigg( \! (J_l^{-T})_{jk}\frac{\partial J_l^{-1}}{\partial x_k}\boldsymbol{e}_i \boldsymbol{x}^T\!+\!J_l^{-1}\boldsymbol{e}_j\boldsymbol{e}_i^T J_l^{-T}\bigg )\\
        &\qquad\quad -(\dot{\mu}\mu^{-1})^{\vee}_i (J_r^{-1}\boldsymbol{e}_i\boldsymbol{x}^T+\boldsymbol{x}\boldsymbol{e}_i^T J_r^{-T})+h_i^c(J_l^{-1}\boldsymbol{e}_i\boldsymbol{x}^T+\boldsymbol{x}\boldsymbol{e}^T_iJ_l^{-T}) \bigg \rangle\\
        &= \bigg\langle \text{sym}\bigg [ \bigg(
\frac{1}{2}\frac{\partial J_l^{-1}}{\partial x_k} (HH^TJ_l^{-T})\boldsymbol{e}_k-J_r^{-1}(\dot{\mu}\mu^{-1})^{\vee}+J_l^{-1}\boldsymbol{h}^c \bigg) \boldsymbol{x}^T \bigg] +J_l^{-1}HH^T J_l^{-T} \bigg \rangle.
    \end{aligned}
\end{equation*}

\subsection{Proof of Theorem \ref{thm:second_order_app}} \label{proof:2nd_prop}
Before expanding terms, we present a result that will be used in the calculation: if $S$ is a symmetric matrix, then
\begin{equation} \label{eq:adS}
    \begin{aligned}
    ad_i S \boldsymbol{e}_i  &\doteq [E_i, (S\boldsymbol{e}_i)^{\wedge}]^{\vee}=[E_i, (S_{ji}\boldsymbol{e}_j)^{\wedge}]^{\vee} \\
    &= [E_i, S_{ji}E_j]^{\vee}=\frac{1}{2}(S_{ji}-S_{ij})[E_i,E_j]^{\vee}=\boldsymbol{0}
    .
    \end{aligned}
\end{equation}
The property $[E_i,E_j]=-[E_j,E_i]$ is used above.

Then, for the mean propagation in (\ref{eq:mean_cov_prop_exact}), expanding all terms to the second-order, we have
\begin{equation*}
\begin{aligned}
    (\dot{\mu}\mu^{-1})^{\vee}
    \approx& \langle I+\frac{1}{12}ad_X^2 \rangle^{-1}\bigg\langle\! \frac{1}{2}\bigg(\!\!-\frac{1}{2}ad_k\!+\!\frac{1}{12}(ad_k ad_X+ad_Xad_k)\bigg)HH^T\\
    &\quad \bigg(I-\frac{1}{2}ad_X+\frac{1}{12}ad^2_X \bigg)^T \boldsymbol{e}_k+\bigg(I-\frac{1}{2}ad_X+\frac{1}{12}ad_X^2\bigg)\\
     & \quad\cdot\bigg( \boldsymbol{h} 
     +\frac{\partial \tilde{\boldsymbol{h}}}{\partial x_i}x_i+\frac{1}{2}\cdot\frac{\partial^2 \tilde{\boldsymbol{h}}}{\partial x_i \partial x_j}x_i x_j \bigg) \bigg \rangle\\ 
\approx& \bigg(I-\frac{1}{12}\langle ad_X^2\rangle\bigg)\bigg \langle \boldsymbol{h}+\frac{1}{12}ad_X^2\boldsymbol{h}-\frac{1}{48}ad_X ad_kHH^Tad_X^T \boldsymbol{e}_k\\
&\, -\frac{1}{48}ad_kHH^T(ad^2_X)^T\!\boldsymbol{e}_k\! -\!\frac{1}{2}x_jad_X  \frac{\partial \tilde{\boldsymbol{h}}}{\partial x_j}+\frac{1}{2}x_ix_j\frac{\partial^2 \tilde{\boldsymbol{h}}}{\partial x_i\partial x_j}\bigg \rangle\\
=&\boldsymbol{h}+\bigg(-\frac{1}{48} ad_i ad_k HH^T ad_j^T \!\boldsymbol{e}_k\!-\!\frac{1}{48}ad_k HH^T ad_j^T ad_i^T \boldsymbol{e}_k \\
&\quad -\frac{1}{2}ad_i \frac{\partial \tilde{\boldsymbol{h}}}{\partial x_j} +\frac{1}{2}\frac{\partial^2  \tilde{\boldsymbol{h}}}{\partial x_i \partial x_j}\bigg)\langle x_ix_j \rangle\\
\end{aligned}
\end{equation*}
where we have used (\ref{eq:adS}) to eliminate two terms from line 2 to line 3.
We also use the approximation $(I+\varepsilon A)^{-1}\approx I-\varepsilon A$.
The covariance propagation is derived in the same way.

\subsection{Proof of Theorem \ref{thm:l_f_connection}} \label{proof:lf_connection}
Assume we have a function $f^*(g)=f(g^{-1})$.
    The left Lie derivative of $f(g)$ is
    \begin{equation*}
        \begin{aligned}
        (E^l_i f)(g,t)&\doteq \frac{d}{d\tau}\bigg|_{\tau=0} f(\exp(-\tau E_i)\circ g,t)\\
        &=\frac{d}{d\tau}\bigg|_{\tau=0} f^*(g^{-1}\circ \exp(\tau E_i),t)\\
        &=(E^r_i f^*)(g^{-1},t).
        \end{aligned}
    \end{equation*}
    By performing the same calculation, we have
    \begin{equation*}
        (E^l_i E^l_j f)(g,t)=(E^r_i E^r_j f^*)(g^{-1},t),
    \end{equation*}
    \begin{equation*}
        E^l_i(h_i f)(g,t)=-E^r_i(h_i^* f^*)(g^{-1},t), 
    \end{equation*}
    \begin{equation*}
        E^l_i E^l_j (H_{ik}H_{kj}^T f)(g,t)=E^r_i E^r_j (H^*_{ik}{H^*_{kj}}^T f^*)(g^{-1},t),
    \end{equation*}
    and
    \begin{equation*}
        \frac{\partial f}{\partial t}(g,t)=\frac{\partial f^*}{\partial t}(g^{-1},t).
    \end{equation*}
    Substituting them back into (\ref{eq:l_FPE}), we have
    the equation which is equivalent to (\ref{eq:r_FPE}).

\subsection{Proof of Theorem \ref{theorem:left_right_statistics}} \label{proof:mean_relation}
For the mean, we have
\begin{equation*}
\begin{aligned}
    \boldsymbol{0}&=\int_G \log^{\vee}(g\circ {\mu_r^*}^{-1})f^*(g)dg\\
    &=\int_G \log^{\vee}(g\circ {\mu^*_r}^{-1})f(g^{-1})dg\\
    &= \int_G \log^{\vee}(g^{-1}\circ {\mu^*_r}^{-1})f(g)dg\\
    &= -\int_G \log^{\vee}(\mu^*_r\circ g)f(g)dg\\
    &= -\int_G \log^{\vee}(\mu_l^{-1} \circ g)f(g)dg.
    \end{aligned}
\end{equation*}
where we have used the property of the bi-invariant Haar measure.
For the covariance, similar calculations yield the result.

\bibliography{IEEEabrv,References_RAL_abbrev}

\begin{thebibliography}{10}
\providecommand{\url}[1]{#1}
\csname url@rmstyle\endcsname
\providecommand{\newblock}{\relax}
\providecommand{\bibinfo}[2]{#2}
\providecommand\BIBentrySTDinterwordspacing{\spaceskip=0pt\relax}
\providecommand\BIBentryALTinterwordstretchfactor{4}
\providecommand\BIBentryALTinterwordspacing{\spaceskip=\fontdimen2\font plus
\BIBentryALTinterwordstretchfactor\fontdimen3\font minus \fontdimen4\font\relax}
\providecommand\BIBforeignlanguage[2]{{%
\expandafter\ifx\csname l@#1\endcsname\relax
\typeout{** WARNING: IEEEtran.bst: No hyphenation pattern has been}%
\typeout{** loaded for the language `#1'. Using the pattern for}%
\typeout{** the default language instead.}%
\else
\language=\csname l@#1\endcsname
\fi
#2}}

\bibitem{barrau2018invariant}
A.~Barrau and S.~Bonnabel, ``{Invariant Kalman filtering},'' \emph{Annu. Rev. Control Robot. Auton. Syst.}, vol.~1, pp. 237--257, 2018.

\bibitem{barrau2016invariant}
------, ``{The invariant extended Kalman filter as a stable observer},'' \emph{IEEE Trans. Autom. Control}, vol.~62, no.~4, pp. 1797--1812, 2016.

\bibitem{phogat2020invariant}
K.~S. Phogat and D.~E. Chang, ``{Invariant extended Kalman filter on matrix Lie groups},'' \emph{Automatica}, vol. 114, p. 108812, 2020.

\bibitem{bourmaud2015continuous}
G.~Bourmaud, R.~M{\'e}gret, M.~Arnaudon, and A.~Giremus, ``{Continuous-discrete extended Kalman filter on matrix Lie groups using concentrated Gaussian distributions},'' \emph{J. Math. Imaging Vis.}, vol.~51, pp. 209--228, 2015.

\bibitem{brossard2017unscented}
M.~Brossard, S.~Bonnabel, and J.-P. Condomines, ``{Unscented Kalman filtering on Lie groups},'' in \emph{IEEE Int. Conf. Intell. Robots Syst.}\hskip 1em plus 0.5em minus 0.4em\relax IEEE, 2017, pp. 2485--2491.

\bibitem{brossard2020code}
M.~Brossard, A.~Barrau, and S.~Bonnabel, ``{A code for unscented Kalman filtering on manifolds (UKF-M)},'' in \emph{IEEE Int. Conf. Robot. Autom.}\hskip 1em plus 0.5em minus 0.4em\relax IEEE, 2020, pp. 5701--5708.

\bibitem{loianno2016visual}
G.~Loianno, M.~Watterson, and V.~Kumar, ``{Visual inertial odometry for quadrotors on SE(3)},'' in \emph{IEEE Int. Conf. Robot. Autom.}\hskip 1em plus 0.5em minus 0.4em\relax IEEE, 2016, pp. 1544--1551.

\bibitem{forbes2017sigma}
J.~R. Forbes and D.~E. Zlotnik, ``{Sigma point Kalman filtering on matrix Lie groups applied to the SLAM problem},'' in \emph{Proc. Int. Conf. Geometric Sci. Inf.}, 2017, pp. 318--328.

\bibitem{sjoberg2021lie}
A.~M. Sj{\o}berg and O.~Egeland, ``{Lie Algebraic unscented Kalman filter for pose estimation},'' \emph{IEEE Trans. Autom. Control}, vol.~67, no.~8, pp. 4300--4307, 2021.

\bibitem{jayaraman2023lietheoretic}
A.~S. Jayaraman, J.~Ye, and G.~S. Chirikjian, ``A {L}ie-theoretic approach to propagating uncertainty jointly in attitude and angular momentum,'' \emph{arXiv preprint arXiv:2309.03112}, 2023.

\bibitem{gardiner1985handbook}
C.~W. Gardiner \emph{et~al.}, \emph{Handbook of stochastic methods}.\hskip 1em plus 0.5em minus 0.4em\relax Springer Berlin, 1985, vol.~3.

\bibitem{biazar2008homotopy}
J.~Biazar, K.~Hosseini, and P.~Gholamin, ``{Homotopy perturbation method Fokker-Planck equation},'' in \emph{Int. Math. Forum}, vol.~3, no.~19, 2008, pp. 945--954.

\bibitem{tatari2007application}
M.~Tatari, M.~Dehghan, and M.~Razzaghi, ``{Application of the Adomian decomposition method for the Fokker-Planck equation},'' \emph{Math. Comput. Model.}, vol.~45, no. 5-6, pp. 639--650, 2007.

\bibitem{biazar2010variational}
J.~Biazar, P.~Gholamin, and K.~Hosseini, ``{Variational iteration method for solving Fokker-Planck equation},'' \emph{J. Frank. Inst.}, vol. 347, no.~7, pp. 1137--1147, 2010.

\bibitem{hemeda2018new}
A.~Hemeda and E.~Eladdad, ``{New iterative methods for solving Fokker-Planck equation},'' \emph{Math. Probl. Eng.}, vol. 2018, pp. 1--9, 2018.

\bibitem{ng2022equivariant}
Y.~Ng, P.~van Goor, T.~Hamel, and R.~Mahony, ``{Equivariant Observers for Second-Order Systems on Matrix Lie Groups},'' \emph{IEEE Trans. Autom. Control}, vol.~68, no.~4, pp. 2468--2474, 2022.

\bibitem{jayaraman2020black}
A.~S. Jayaraman, D.~Campolo, and G.~S. Chirikjian, ``{Black-Scholes theory and diffusion processes on the cotangent bundle of the affine group},'' \emph{Entropy}, vol.~22, no.~4, p. 455, 2020.

\bibitem{jayaraman2023inertial}
A.~S. Jayaraman, J.~Ye, and G.~S. Chirikjian, ``{On the inertial rotational Brownian motion of arbitrarily shaped particles},'' \emph{arXiv preprint arXiv:2303.07021}, 2023.

\bibitem{chirikjian2016harmonic}
G.~S. Chirikjian and A.~B. Kyatkin, \emph{Harmonic Analysis for Engineers and Applied Scientists: Updated and Expanded Edition}.\hskip 1em plus 0.5em minus 0.4em\relax Courier Dover Publications, 2016.

\bibitem{chirikjian2011stochastic}
G.~S. Chirikjian, \emph{Stochastic Models, Information Theory, and Lie Groups, Volume 2: Analytic Methods and Modern Applications}.\hskip 1em plus 0.5em minus 0.4em\relax Springer Science \& Business Media, 2011, vol.~2.

\bibitem{park1991optimal}
F.~C. Park, \emph{The optimal kinematic design of mechanisms}.\hskip 1em plus 0.5em minus 0.4em\relax Harvard Univ., 1991.

\bibitem{wang2006error}
Y.~Wang and G.~S. Chirikjian, ``{Error propagation on the Euclidean group with applications to manipulator kinematics},'' \emph{IEEE Transactions on Robotics}, vol.~22, no.~4, pp. 591--602, 2006.

\bibitem{long2013banana}
A.~W. Long, K.~C. Wolfe, M.~J. Mashner, and G.~S. Chirikjian, ``{The banana distribution is Gaussian: A localization study with exponential coordinates},'' in \emph{Proc. Robot. Sci. and Syst.}, vol. 265, 2012.

\bibitem{park2010path}
W.~Park, Y.~Wang, and G.~S. Chirikjian, ``The path-of-probability algorithm for steering and feedback control of flexible needles,'' \emph{The International journal of robotics research}, vol.~29, no.~7, pp. 813--830, 2010.

\bibitem{barfoot2014associating}
T.~D. Barfoot and P.~T. Furgale, ``Associating uncertainty with three-dimensional poses for use in estimation problems,'' \emph{IEEE Trans. Robotics}, vol.~30, no.~3, pp. 679--693, 2014.

\bibitem{adurthi2018conjugate}
N.~Adurthi, P.~Singla, and T.~Singh, ``Conjugate unscented transformation: Applications to estimation and control,'' \emph{J. Dyn. Syst. Meas. Control}, vol. 140, no.~3, 2018.

\bibitem{kuehn2016moment}
C.~Kuehn, ``Moment closure—a brief review,'' \emph{Control of self-organizing nonlinear systems}, pp. 253--271, 2016.

\bibitem{bullo1995proportional}
F.~Bullo and R.~M. Murray, ``{Proportional derivative (PD) control on the Euclidean group},'' in \emph{Proc. Eur. Control Conf.}, 1995, pp. 1091--1097.

\bibitem{sarkka2007unscented}
S.~Sarkka, ``{On unscented Kalman filtering for state estimation of continuous-time nonlinear systems},'' \emph{IEEE Trans. Autom. Control}, vol.~52, no.~9, pp. 1631--1641, 2007.

\bibitem{fujikawa1995spacecraft}
S.~J. Fujikawa and D.~F. Zimbelman, ``{Spacecraft attitude determination by Kalman filtering of Global Positioning System signals},'' \emph{J. Guid. Control Dyn.}, vol.~18, no.~6, pp. 1365--1371, 1995.

\bibitem{axelrad1996spacecraft}
P.~Axelrad and L.~M. Ward, ``Spacecraft attitude estimation using the global positioning system-methodology and results for radcal,'' \emph{J. Guid. Control Dyn.}, vol.~19, no.~6, pp. 1201--1209, 1996.

\bibitem{psiaki2009generalized}
M.~L. Psiaki, ``{Generalized Wahba problems for spinning spacecraft attitude and rate determination},'' \emph{J. Astronaut. Sci.}, vol.~57, no. 1-2, pp. 73--92, 2009.

\bibitem{humphreys2005magnetometer}
T.~E. Humphreys, M.~L. Psiaki, E.~M. Klatt, S.~P. Powell, and P.~M. Kintner~Jr, ``Magnetometer-based attitude and rate estimation for spacecraft with wire booms,'' \emph{J. Guid. Control Dyn.}, vol.~28, no.~4, pp. 584--593, 2005.

\bibitem{sanyal2008global}
A.~K. Sanyal, T.~Lee, M.~Leok, and N.~H. McClamroch, ``Global optimal attitude estimation using uncertainty ellipsoids,'' \emph{Syst. Control Lett.}, vol.~57, no.~3, pp. 236--245, 2008.

\bibitem{ma2014magnetometer}
H.~Ma and S.~Xu, ``Magnetometer-only attitude and angular velocity filtering estimation for attitude changing spacecraft,'' \emph{Acta Astronaut.}, vol. 102, pp. 89--102, 2014.

\bibitem{srivastava2021attitude}
R.~Srivastava, R.~Sah, and K.~Das, ``Attitude and in-orbit residual magnetic moment estimation of small satellites using only magnetometer,'' \emph{arXiv preprint arXiv:2107.09257}, 2021.

\bibitem{julier1995new}
S.~J. Julier, J.~K. Uhlmann, and H.~F. Durrant-Whyte, ``A new approach for filtering nonlinear systems,'' in \emph{Proc. Amer. Control Conf.}, vol.~3.\hskip 1em plus 0.5em minus 0.4em\relax IEEE, 1995, pp. 1628--1632.

\bibitem{sarkka2023bayesian}
S.~S{\"a}rkk{\"a} and L.~Svensson, \emph{Bayesian filtering and smoothing}.\hskip 1em plus 0.5em minus 0.4em\relax Cambridge Univ. Press, 2023, vol.~17.

\bibitem{roberts2012modify}
A.~Roberts, ``Modify the improved euler scheme to integrate stochastic differential equations,'' \emph{arXiv preprint arXiv:1210.0933}, 2012.

\bibitem{chiuso2000monte}
A.~Chiuso and S.~Soatto, ``{Monte Carlo filtering on Lie groups},'' in \emph{Proc. 39th IEEE Conf. Decision Control}, vol.~1.\hskip 1em plus 0.5em minus 0.4em\relax IEEE, 2000, pp. 304--309.

\bibitem{marjanovic2016engineer}
G.~Marjanovic and V.~Solo, ``{An engineer's guide to particle filtering on matrix Lie groups},'' in \emph{Proc. IEEE Int. Conf. Acoust. Speech Signal Process}.\hskip 1em plus 0.5em minus 0.4em\relax IEEE, 2016, pp. 3969--3973.

\bibitem{zhang2016feedback}
C.~Zhang, A.~Taghvaei, and P.~G. Mehta, ``{Feedback particle filter on matrix Lie groups},'' in \emph{Proc. Amer. Control Conf.}\hskip 1em plus 0.5em minus 0.4em\relax IEEE, 2016, pp. 2723--2728.

\bibitem{bangura2012nonlinear}
M.~Bangura, R.~Mahony, \emph{et~al.}, ``Nonlinear dynamic modeling for high performance control of a quadrotor,'' in \emph{Proc. Australasian Conf. Robot. Autom.}\hskip 1em plus 0.5em minus 0.4em\relax IEEE, 2012, pp. 1--10.

\bibitem{lee2010geometric}
T.~Lee, M.~Leok, and N.~H. McClamroch, ``{Geometric tracking control of a quadrotor UAV on SE (3)},'' in \emph{Proc. 49th IEEE Conf. Decision Control}.\hskip 1em plus 0.5em minus 0.4em\relax IEEE, 2010, pp. 5420--5425.

\bibitem{chirikjian2009stochastic}
G.~S. Chirikjian, \emph{Stochastic Models, Information Theory, and Lie Groups, Volume 1: Classical Results and Geometric Methods}.\hskip 1em plus 0.5em minus 0.4em\relax Springer Science \& Business Media, 2009.

\end{thebibliography}

\renewenvironment{IEEEbiography}[1]
  {\IEEEbiographynophoto{#1}}
  {\endIEEEbiographynophoto}
\vspace{-10 mm}
\begin{IEEEbiography}
{Jikai Ye} received his Bachelor’s degree in
Theoretical and Applied Mechanics from Peking University, China, in 2021. 
He became a Ph.D. student at the Department of Mechanical Engineering at the National University of Singapore under the supervision of Professor Gregory Chirikjian in 2021. 
His research interests include the application of group theory in state estimation and robot manipulation.
\vspace{-10 mm}
\end{IEEEbiography}

\begin{IEEEbiography}
{Amitesh S. Jayaraman} received his Bachelor of Engineering degree in Mechanical Engineering from the National University of Singapore in 2020. Since 2020, he has been a Ph.D. student at The NanoEnergy lab in Stanford University, under the supervision of Professor Hai Wang. His current research interests are in the areas of molecular and quantum transport, as well as non-equilibrium processes. Prior to his Ph.D. work, he spent his final year as an undergraduate working with Professor Gregory Chirikjian on a range of topics involving stochastic processes on Lie groups, biomolecular reconstruction and mean/covariance propagation techniques in the solution of partial differential equations.  
\vspace{-10 mm}
\end{IEEEbiography}

\begin{IEEEbiography}
{Gregory S. Chirikjian} (Fellow, IEEE) received undergraduate degrees from Johns Hopkins University in 1988, and a Ph.D. degree from the California Institute of Technology, Pasadena, in 1992. 

From 1992 until 2021, he served on the faculty of the Department of Mechanical Engineering at Johns Hopkins University, attaining the rank of full professor in 2001. Additionally, from 2004-2007, he served as department chair. Starting in January 2019, he moved to the National University of Singapore, where he served as Head of the Mechanical Engineering Department until June 2023, during which time he has hired 14 new professors. 
He is the author of more than 250 journal and conference papers and the primary author of three books, including Engineering Applications of Noncommutative Harmonic Analysis (2001) and Stochastic Models, Information Theory, and Lie Groups, Vols. 1+2. (2009, 2011). In 2016, an expanded edition of his 2001 book was published as a Dover book under a new title, Harmonic Analysis for Engineers and Applied Scientists.
His research interests include robotics, applications of group theory in state estimation, information-theoretic inequalities, and applied mathematics more broadly. 

Prof. Chirikjian is a 1993 National Science Foundation Young Investigator and a 1994 Presidential Faculty Fellow. In 2008 he became a fellow of the ASME and in 2010 he became a fellow of the IEEE. From 2014-15, he served as a program director for the US National Robotics Initiative, which included responsibilities in the Robust Intelligence cluster in the Information and Intelligent Systems Division of CISE at NSF.

\end{IEEEbiography}

\end{document}